\documentclass[twocolumn,10pt]{asme2ej}
\usepackage{asmewide}
\usepackage{graphicx} 
\usepackage{hyperref}   
\hypersetup{
	colorlinks=true,
	linkcolor=blue,
	citecolor=blue,
	urlcolor=blue,
}
\usepackage[square,numbers]{natbib}
\newcommand{\pd}{p}

\usepackage{array}   
\usepackage{booktabs} 
\usepackage{caption} 
\usepackage{siunitx} 
\usepackage{epstopdf} 
\usepackage{graphicx}

\usepackage{subcaption}
\usepackage{amsmath,amssymb,amsfonts,mathtools,nccmath,bm, tikz}
\allowdisplaybreaks
\newtheorem{theorem}{Theorem}[section]
\usepackage{multicol}
\newtheorem{definition}{Definition}[section] 
\newtheorem{assumption}{Assumption}[section]
\newtheorem{remark}{Remark}[section]
\newtheorem{proposition}[theorem]{Proposition}
\newtheorem{corollary}[theorem]{Corollary}
\newtheorem{lemma}[theorem]{Lemma}
\usepackage{algorithmic}
\usepackage{textcomp}
\usepackage{cuted}
\usepackage{float}
\usepackage{xcolor}
\usepackage{physics}
\usepackage{setspace}
\def\BibTeX{{\rm B\kern-.05em{\sc i\kern-.025em b}\kern-.08em
    T\kern-.1667em\lower.7ex\hbox{E}\kern-.125emX}}
\allowdisplaybreaks

\usepackage{tabstackengine}
\usetikzlibrary{shapes,shapes.geometric,arrows.meta,fit,calc,positioning,automata,patterns}
\usetikzlibrary{tikzmark,calc,decorations.pathreplacing}
\tikzset{
    node_style/.style={
        draw, fill=green!50, text=black, regular polygon, regular polygon sides=8,
        minimum size=0.4cm, font=\scriptsize\bfseries
    },
    edge_style/.style={draw, ->, thick, purple},
    edge_style_1/.style={draw, <->, thick, purple},
}

\title{Distributed Leader-Follower Consensus for Uncertain Multiagent Systems with Time-Triggered Switching of the Communication Network \thanks{Paper accepted and will be presented at the 2025 Modeling, Estimation, and Control Conference (MECC 2025), Pittsburgh, PA, Oct. 5–8. Paper No. MECC2025-60.}}

\author{Armel Koulong \thanks{Corresponding author.} 
    \affiliation{
	 Member of ASME\\
     Multi-Modal Multi-Agent Control Lab\\
	Department of Mechanical Engineering\\
	University of Alabama\\
	Tuscaloosa, Alabama 35401\\
    Email: akoulongzoyem@crimson.ua.edu
    }	
}

\author{Ali Pakniyat 
    \affiliation{ Member of ASME\\
	Multi-Modal Multi-Agent Control Lab\\
	Department of Mechanical Engineering\\
	University of Alabama\\
	Tuscaloosa, Alabama 35401\\
    Email: apakniyat@ua.edu
    }
}

\begin{document}

\maketitle    

\begin{abstract}
{\it 
A distributed adaptive control strategy is developed for heterogeneous multiagent systems in nonlinear Brunovsky form with \({\pd}\)-dimensional $n^{\text{th}}$-order dynamics, operating under time-triggered switching communication topologies. The approach uses repulsive potential functions to ensure agent-agent and obstacle safety, while neural network estimators compensate for system uncertainties and disturbances. A high-order control barrier function framework is then employed to certify the positive invariance of the safe sets and the boundedness of the proposed control inputs. The resulting distributed control and adaptive laws, together with dwell-time requirements for topology transitions, achieve leader-following consensus. This integrated design provides synchronized formation and robust disturbance rejection in evolving network configurations, and its effectiveness is demonstrated through numerical simulations.

}
\end{abstract}

\begin{nomenclature}
\entry{$a_{ij}^\sigma$}{Adjacency weight from agent $j$ to $i$ in mode $\sigma$.}
\entry{$A^\sigma$}{Weighted adjacency matrix of $\mathcal G^\sigma$.}
\entry{$D^\sigma$}{In-degree matrix $\mathrm{diag}\!\big(\sum_j a_{ij}^\sigma\big)$.}
\entry{$L^\sigma$}{Graph Laplacian of $\mathcal G^\sigma$.}
\entry{$\mathbb{R}$}{Set of real numbers.}
\entry{$\mathbb{R}^{\pd}$}{Set of real vectors of dimension $\pd$.}
\entry{$\lambda_{ij}$}{Synchronization-error weight between agents $i$ and $j$.}
\entry{$\psi_{ij}^1$}{Minimum separation (collision-avoidance threshold) between agents $i$ and $j$.}
\entry{$\psi_{i0}^1$}{Minimum separation (collision-avoidance threshold) between agent $i$ and the leader.}
\entry{$R$}{Outer safety radius around an obstacle (threshold distance).}
\entry{$\mathfrak{L}$}{Inner exclusion radius around an obstacle (threshold distance).}
\entry{$\varpi$}{Inter-agent/leader repulsion strength.}
\entry{$\kappa_i$}{Scalar tuning gain in NN update laws.}
\entry{$\kappa_0$}{Scalar tuning gain in NN update laws.}
\entry{$\kappa_{iw}$}{Scalar tuning gain in NN update laws.}
\entry{$\Gamma^0$}{Gain matrix scaling obstacle/collision-avoidance inputs.}
\entry{$\Gamma^1$}{Gain matrix scaling obstacle/collision-avoidance inputs.}
\entry{$\Gamma^2$}{Gain matrix scaling obstacle/collision-avoidance inputs.}
\entry{$\mathrm{tr}\{\cdot\}$}{Trace of a matrix.}
\entry{$t \in \mathbb{R}_{\ge 0}$}{Continuous time.}
\entry{$|\cdot|$}{Absolute value.}
\entry{$\|\cdot\|$}{Euclidean (vector 2-)norm.}
\entry{$\|\cdot\|_G$}{Frobenius norm for matrices.}
\entry{$P^{-1}$}{Inverse of matrix $P$.}
\entry{$t_s$}{Switching instant.}
\entry{$t_{s-}$}{Time immediately before the switching instant.}
\entry{$\alpha(\cdot)$}{Singular values.}
\entry{$\bar{\alpha}(\cdot)$}{Maximal singular values.}
\entry{$\underline{\alpha}(\cdot)$}{Minimal singular values.}
\entry{$m_{ij}$}{Repulsive potential for agent–agent interaction.}
\entry{$m_{i0}$}{Repulsive potential for agent–leader interaction.}
\entry{$m_{i{\mathfrak{c}}}$}{Repulsive potential for agent–obstacle interaction.}
\entry{$\mathcal{G}^\sigma=(\mathcal{V},\mathcal{E}^\sigma)$}{Communication graph under mode $\sigma$.}
\entry{$\bar{\mathcal G}^\sigma$}{Augmented graph including leader node $0$.}
\entry{$\hat\theta_i$}{Estimated neural weights for agent $i$.}
\entry{$\hat\theta_{iw}$}{Estimated neural weights for the disturbance on agent $i$.}
\entry{$\hat\theta_0$}{Estimated neural weights for the leader.}
\entry{$\delta_{ij}$}{Positive lower bound on the initial separation between agents $i$ and $j$.}
\entry{$\delta_{i0}$}{Positive lower bound on the initial separation between agent $i$ and the leader.}
\entry{$X_{\mathrm{init}}$}{Upper bound on the norm of any agent’s initial state.}
\entry{$f_{\mathrm{init}}(\cdot)$}{Continuous function bounding each $f_i(x)$ on a compact set.}
\entry{$f_0(x_0)$}{Continuous function bounding the leader’s dynamics $f_0(x_0)$.}
\entry{$w_{\mathrm{init}}$}{Upper bound on the norm of the disturbance.}
\entry{$\mathbb{R}^{(N \cdot {\pd}) \times (N \cdot {\pd})}$}{Set of real matrices with $(N \cdot {\pd})$ rows and $(N \cdot {\pd})$ columns.}
\entry{$G_i \in \mathbb{R}^{q_i \times q_i}$}{Positive definite adaptation-gain matrix for the NN of agent-dynamics.}
\entry{$G_{i0} \in \mathbb{R}^{q_0 \times q_0}$}{Positive definite adaptation-gain matrix for the NN of leader-dynamics.}
\entry{$G_{iw} \in \mathbb{R}^{q_{iw} \times q_{iw}}$}{Positive definite adaptation-gain matrix for the NN of disturbance model.}
\end{nomenclature}

\section{Introduction}
The cooperative control of multi-agent systems (MAS) is essential for applications such as unmanned aerial vehicle (UAV) coordination, autonomous vehicle platooning, and smart grids. In these systems, agents interact only with a limited number of other agents, creating a graph-theoretic structure that complicates achieving global consensus \citep{lewis2013,1431045}. This challenge is further amplified in heterogeneous uncertain MAS, where agents with different forms of high-order non-linear dynamics are subject to external disturbances \citep{Khoo2009}. Moreover, adapting to unknown parameters in real-time further complicates the consensus process, as agents must continuously update their internal models based on limited interactions.

Despite advancements in the field, significant gaps remain in addressing adaptive control for leader-following consensus with obstacle and collision avoidance in nonlinear systems. Foundational works have addressed these challenges in isolation, but no single framework provides an integrated solution. For instance, adaptive tracking control methods \citep{lewis2013,Hong2006} and consensus approaches for switching communication topologies \citep{Qin2011} often neglect physical constraints like collision avoidance. Furthermore, while the adaptive tracking method in \citep{Hong2006} effectively handles switching topologies, its formulation is limited to linear agent dynamics and does not extend to heterogeneous nonlinear MAS.  Conversely, research that does consider safety—such as work on local interaction rules \citep{Olfati2006}, optimal control for consensus \citep{Wang2011}, and geometric formation conditions \citep{Lin2005}—often overlooks the complexities of nonlinear dynamics, adaptation, and network switching. Stability analyses for time-dependent links \citep{Moreau2005} and switching topologies \citep{Shi2009} highlight network challenges but only partially address physical constraints and heterogeneous dynamics, while ignoring the need for real-time nonlinear estimation.

Recent research in the past few years has significantly sharpened the landscape for adaptive multi-agent systems under switching communication topologies. Some studies explore consensus for unknown nonlinear MAS with communication noise under Markov switching, using Radial Basis Function networks and stochastic stability conditions~\citep{Guo2024Mathematics}; these efforts focus on state consensus without addressing formation, safety, or average dwell time (ADT) considerations. Other work investigates positive consensus for linear MAS with ADT switching, establishing infimum dwell times and positivity-preserving conditions for consensus but without encompassing nonlinearity, heterogeneity, adaptations to uncertainty, or safety aspects~\citep{Cao2021JFI}.

Event-triggered designs have also been proposed; some researchers have introduced decentralized protocols with dynamic event triggering and guaranteed minimum inter-event times to improve robustness and communication efficiency~\citep{Wu2022SIAM}, and others have developed neural adaptive event-triggered consensus methods for unknown nonlinear MAS on switched or Markovian networks~\citep{Luo2023Math}. However, these references do not study safety guarantees or explicit dwell time constraints, as their approaches prioritize consensus and communication.

In the context of formation control under switching topologies, research has often been tailored to address the requirements of specific applications. For instance, distributed adaptive robust schemes have been developed and experimentally validated for quadrotor UAVs, demonstrating effectiveness in outdoor environments~\citep{Zhao2023Aerospace}. Similarly, other studies have targeted formation control for fractional-order multi-agent systems with actuator faults~\citep{Li2024FractalFract}, again focusing on the specific characteristics pertinent to those systems. Despite these advances in application-oriented domains, a comprehensive and unified framework that integrates heterogeneous higher-order agents, obstacle and collision avoidance via potential functions, distributed neural adaptation, and provable minimum average dwell time guarantees for leader-follower formation tracking remains unaddressed in the literature.

In past work of the authors \citep{koulong2024}, we introduced an adaptive control strategy for multiagent systems with fixed communication topology, focusing specifically on agents operating in a one-dimensional physical space. Due to the formulation adopted therein, the approach was limited to modeling collision avoidance scenarios restricted to one-dimensional interactions. 
In this paper, we extend the framework presented in \citep{lewis2013, koulong2024} for heterogeneous multiagent systems in the general Brunovsky form with \mbox{\({\pd}\)-dimensional} nonlinear $n^{\text{th}}$-order dynamics operating under time-triggered switching communication topologies. A key distinguishing feature of this paper is the presentation of stability and convergence proofs that clearly establish conditions for global asymptotic stability, synchronization, and performance under topology switching.

The proposed methodology effectively incorporates both collision and obstacle avoidance under time-triggered switching communication topologies. Specifically, we integrate a distributed adaptive control architecture with a leader-follower formation structure, enabling heterogeneous multiagent systems to achieve consensus and coordination despite evolving network configurations. 
To handle uncertainties arising from unknown agent dynamics and external disturbances, we employ neural network-based estimation techniques coupled with real-time adaptive tuning laws which continuously update the control parameters to ensure reliable performance under significant model uncertainties. Additionally, we incorporate potential functions to guarantee collision-free trajectories and obstacle avoidance while maintaining cohesive formations.
In particular, we determine minimum average dwell-time conditions based on the neural network tuning parameters and the control laws, which together ensure leader-follower synchronization, collision/obstacle avoidance, and global asymptotic stability under switching communication topologies.

Building upon foundational work in single-agent switched systems~\citep{hespanha2007networked, Liberzon2003}, adaptive tracking~\citep{lewis2013}, and various multi-agent consensus formulations for agents with linear dynamics—including consensus with obstacle and collision avoidance using local information under fixed topologies~\citep{Wang2011}, as well as leader–follower consensus under variable topologies~\citep{Hong2006}—this paper presents a unified control framework for heterogeneous multi-agent systems with high-order Brunovsky nonlinear dynamics subject to switching communication topologies.  To the best of our knowledge, no prior work provides a unified, decentralized design for such nonlinear systems that combines distributed linear-in-parameters (LIP)–neural network (NN) adaptation, time-triggered switching communication topologies, and potential-based safety assurance, together with a composite Lyapunov proof yielding explicit minimum average-dwell-time bounds for leader–follower formation tracking in heterogeneous \mbox{$p$-dimensional}, $n$-th-order Brunovsky systems. 

The main contributions of this paper are:
\begin{enumerate}
    \item A unified, decentralized control architecture for heterogeneous, nonlinear multi-agent systems (MAS) that simultaneously addresses neural adaptations to uncertain dynamics, repulsion-based safety assurance for collision and obstacle avoidance, and leader–follower formation consensus under switching communication topologies.
    \item A composite Lyapunov-based analysis that establishes explicit, topology- and parameter-dependent minimum average-dwell-time conditions, guaranteeing the simultaneous achievement of all control objectives (formation, safety, and adaptation).
    \item Mathematical proofs of global asymptotic stability and safety, which ensure bounded control inputs and predictable performance with guaranteed error bounds that are explicitly linked to tunable design parameters.
    \item A critical synthesis of foundational literature that overcomes the limitations of prior work—which addressed challenges like topology switching, safety assurance and agent nonlinearity only in isolation—to deliver a unified, provably correct, and decentralized framework that simultaneously guarantees safety, stability, and adaptation.
\end{enumerate}
The remainder of the paper is structured as follows:
Section~\ref{ProblemFormulation} defines MAS dynamics, switching communication topology and the leader-following framework. Section~\ref{Methodology} presents an adaptive control protocol with neural network-based learning for managing high dimensional non-linear dynamics and
disturbances in a time-triggered switching. Section~\ref{MainResult} proves control law stability and convergence. Section~\ref{NUMERICALEXAMPLE} demonstrates the strategy via a numerical leader-follower example with switching communication
topology. Section~\ref{CONCLUSION} summarizes contributions and potential future work.

\section{Problem Formulation} \label{ProblemFormulation}

We denote by $\mathcal{N} = \{1,2,\ldots,N\}$ the set of $N$ follower agents.  Each follower’s full state is
\( x_i = [ x_i^1,  x_i^2, \cdots,  x_i^n ]^\top \in \mathbb{R}^{ n {\cdot {\pd} }},\) with \(x_i^k\in\mathbb{R}^p\), \(k=1,2,\ldots,n,\) where $p$ is the dimension of the physical space and $x_i^k$ corresponds to the $(k-1)^{\text{th}}$ derivative of the agent’s position (e.g., for two dimensional motions, $p=2$, thus $x_{i}^{k} = \bigl[x_{i}^{k,1},\,x_{i}^{k,2}\bigr]^\top\in\mathbb R^2$). The leader’s state $x_0\in\mathbb{R}^{n \cdot p}$ is partitioned similarly. 

The dynamics of the \(i^\text{th}\) follower agent is described using the non-linear Brunovsky form as:
\begin{equation}\label{planta}
\begin{aligned}
\dot{x}_{i}^{1} &= x_{i}^{2}\\
\dot{x}_{i}^{2} &= x_{i}^{3} \\
&~\,\vdots \\
\dot{x}_{i}^{n} &= f_{i}(x_{i}) + u_{i} + w_{i}(t)
\end{aligned}
\end{equation}
where \(u_{i} \in \mathbb{R}^{\pd} \) is the control input of agent \(i\), and \(w_{i}(t) \in \mathbb{R}^{\pd} \) denotes a bounded unknown time-varying disturbance affecting agent \(i\).

We assume that the agent-specific functions \( f_i(x_i) : \mathbb{R}^{n \cdot { {\pd} }} \rightarrow \mathbb{R}^{\pd} \) are unknown to the agents, but these functions are considered to be locally Lipschitz in \( x_i \). For convenience in the Lyapunov stability analysis, we further adopt $f_i(0) = 0$ for all $i \in \mathcal{N}$; this assumption can be relaxed with appropriate changes in the analytical arguments (e.g., via invariance principles).

While the governing dynamics of Eqn.~\eqref{planta} of the agents are decoupled, their interactions become coupled through collision and obstacle avoidance, as discussed later in this~section.

In discussions regarding the collective behavior of following agents, we utilize the Kronecker product \citep{Wang2011} to collectively represent the dynamics of Eqn.~\eqref{planta} of all agents in the global form:
\begin{equation}\label{plantag}
\begin{aligned}
\dot{x}^{1} &= x^{2}\\[-5pt]
&~\,\vdots \\
\dot{x}^{n} &= f(x) + u  + w 
\end{aligned}
\end{equation}
where \( x^{k} = [x_{1}^{k}; x_{2}^{k}; \ldots; x_{N}^{k}] \in \mathbb{R}^{N \cdot {\pd}} \) denotes the concatenated \( k \)-th state vectors of all follower agents, \( u = [u_{1}; u_{2}; \ldots; u_{N}] \in \mathbb{R}^{N \cdot {\pd}} \) denotes their control inputs, \( w = [w_{1}; w_{2}; \ldots; w_{N}] \in \mathbb{R}^{N \cdot {\pd}} \) denotes unknown bounded disturbances, the unknown function \( f(x) = [f_{1}(x_{1}); f_{2}(x_{2}); \ldots; f_{N}(x_{N})] \in \mathbb{R}^{N \cdot {\pd}} \) collectively represents the non-linear dynamics of all \( N \) follower agents.

In the leader-follower scenario considered in this paper, the leader agent, denoted by subscript $0$, generates a reference trajectory for follower agents, but this reference trajectory is a priori unknown to all follower agents, and the only information about the reference trajectory available to follower agents is the leader's state vector \(x_0 \equiv x_0(t) \in \mathbb{R}^{n \cdot {\pd}}\) at the current time $t \in [t_0,\infty)$. The leader also determines a time-triggered switching rule for the communication topology as discussed later.

The time-varying dynamics of the leader agent is:
\begin{equation}\label{plantref}
\begin{aligned}
\dot{x}_{0}^{1}(t) &= x_{0}^{2}(t) \\[-5pt]
& ~\,\vdots \\
\dot{x}_{0}^{n}(t) &= f_{0}(x_{0},t)
\end{aligned}
\end{equation}
where \( x_{0}^{k} \in \mathbb{R}^{\pd} \) is the \(k\)-th partition of the leader agent's state, and \(x_0 = [{x}_{0}^1,{x}_{0}^2,\ldots,{x}_{0}^n] \in \mathbb{R}^{n \cdot {\pd}} \) is its state vector. 

The unknown function \( f_0(x_{0}, t) : [0, \infty) \times \mathbb{R}^{n \cdot {\pd}} \rightarrow \mathbb{R}^{\pd} \) in the leader dynamics is considered to be 
piecewise continous in~\(t\) and locally Lipschitz in \(x_0\), with \( f_i(0,t) = 0 \) for all $t \ge 0$, and all \( x_0 \in \mathbb{R}^{n \cdot {\pd}} \). The governing dynamics Eqn~\eqref{plantref} of the leader agent is private to the leader agent, and its current state value $x_0$ is shared only with a subset of the follower agents, serving as a reference for the formation control objectives discussed further below.
The state $x_0$ is explicitly said to serve as a reference for the followers’ formation control objectives. This suggests the formation is not static, but rather intended to track or align with the moving leader trajectory. The condition $f_0(0, t) = 0$ ensures that if the leader starts at the origin, it could stay there. But there is no indication that it actually does — the general form of the dynamics and the language used implies that motion is expected. Given that the leader’s dynamics are general and time-varying, it serves as a formation reference.  This is a dynamic tracking formation problem. The followers are expected to maintain a formation  with respect to a moving leader state $x_0$, not converge to a fixed spatial pattern at the origin.

In the consensus problem, it is desired that 
\begin{align} 
    \lim_{t \to \infty} [(x^{1}_{i}-\psi_{i}^1) - (x^{1}_{0}- \psi_{0}^1)] &= 0 \label{lam} \\
    \lim_{t \to \infty} [(x^{1}_{i}-\psi_{i}^1) - (x^{1}_{j}- \psi_{j}^1)] &= 0 \label{la}
\end{align}
for each $i \in \mathcal{N}$ and for all $j  \in \mathcal{N}$, $j \ne i$, where \(\psi_{i}^1\), \(\psi_{j}^1\) and \(\psi_{0}^1 \in \mathbb{R}^{\pd} \) denote desired position offsets for agents \(i\), \(j\), and the leader agent \( 0 \), respectively. Following standard consensus/formation definitions (e.g., \citep{WeiRen, Hong2006}), we adopt the desired asymptotic relations. Each agent aims to achieve a state offset from the leader or other agents by \( \psi_i^k \), ensuring coordinated movement while maintaining formation. However, such a strong level of connectivity which requires communication among all agents is often not feasible in networked control systems due to practical constraints such as limited bandwidth, packet loss, and delays (e.g., see \citep{hespanha2007networked}). 

The \( k^\text{th} \) order tracking error is defined as \(\delta_{i}^k = (x_{i}^k - \psi_{i}^k) - (x_{0}^k - \psi_{0}^k) \in \mathbb{R}^{{\pd}}\), and globally as 
\begin{equation}\label{delta_errors}
\delta^{k} = [\delta_{1}^k, \delta_{2}^k, \ldots, \delta_{N}^k]^{\top} \in \mathbb{R}^{N{\cdot \pd}}.
\end{equation}

\begin{definition}
The tracking error \(\delta^{k}\) is \textit{cooperatively uniformly ultimately bounded} (CUUB) \citep{lewis2013} if there exists a compact set \(\varsigma^k \subset \mathbb{R}^{N \cdot {\pd}}\) with \(\{0\} \subset \varsigma^k\) such that for any initial condition \(\delta^k(t_0) \in \varsigma^k\), there exists \(B^k \in \mathbb{R}_{>0}\) and \(T_k \in \mathbb{R}_{>0}\) which yields \(||\delta^k(t)|| \le B^k\) for all \(t \ge t_0 + T_k\). \end{definition}

In this paper, we seek a distributed consensus policy where, in addition to its own state \( x_i \), each agent determines its input \( u_i \) based solely on the states of a limited number of other agents, which may or may not include the leader agent. Furthermore, each agent must also account for obstacle and collision avoidance.

To shape motions for collision/obstacle avoidance we employ repulsive potential terms in the spirit of \citep{Wang2011}.

\paragraph{Agent–agent repulsion.}
For $x_i^1,x_j^1\in\mathbb R^{\pd}$ and $\psi_{ij}^1>0$,
\begin{equation}
m_{ij} \;:=\;
\begin{cases}
0, & \|x_i^1-x_j^1\|>\psi_{ij}^1,\\[3pt]
\dfrac{\varpi}{\|x_i^1-x_j^1\|}, & \|x_i^1-x_j^1\|\le\psi_{ij}^1,
\end{cases}
\label{col1}
\end{equation}

\paragraph{Agent–leader repulsion.}
For $x_0^1\in\mathbb R^{\pd}$ and $\psi_{i0}^1>0$,
\begin{equation}
m_{i0} \;:=\;
\begin{cases}
0, & \|x_i^1-x_0^1\|>\psi_{i0}^1,\\[3pt]
\dfrac{\varpi}{\|x_i^1-x_0^1\|}, & \|x_i^1-x_0^1\|\le\psi_{i0}^1,
\end{cases}
\label{col2}
\end{equation}

\paragraph{Agent–obstacle repulsion.}
Let the obstacle index set be $\mathcal B=\{1,\dots,\Xi\}$, with center $O_{\mathfrak{c}}\in\mathbb R^{\pd}$, detection radius $R>0$, and inner radius $\mathfrak L\in(0,R)$ for ${\mathfrak{c}}\in\mathcal B$. 
\begin{equation} \label{obs1}
m_{i{\mathfrak{c}}} :=
    \begin{cases}
        \quad 0, &  \|x_i^1-O_{\mathfrak{c}}\|>R,  \\
        \Bigg[\frac{R^2 - ||x_i^1 - O_{\mathfrak{c}}||^2}{\left|\left|x_i^1 - O_{\mathfrak{c}}\right|\right|^2 - \mathfrak{L}^2}\Bigg]^2,   & \mathfrak L<\|x_i^1-O_{\mathfrak{c}}\|\le R.
    \end{cases}
\end{equation}
The formulation for obstacle avoidance between the leader agent and obstacle is similar, with \( x_i^1 \) replaced by \( x_0^1 \).
These scalars ($m_{ij}, m_{i0}, m_{i{\mathfrak{c}}}$) enter the control law (repulsion channels), while safety is guaranteed separately via High-Order Control Barrier Functions (HOCBFs) (forward invariance of the safe sets) defined on
$h_{ij}=\|x_i^1-x_j^1\|-\psi_{ij}^1$, $h_{i0}=\|x_i^1-x_0^1\|-\psi_{i0}^1$, and $h_{i{\mathfrak{c}}}=\|x_i^1-O_{\mathfrak{c}}\|- R$ (see Appendix.~\ref{app:HOCBF}). In these expressions, \( \varpi \in \mathbb{R}_{>0} \) is a positive scalar adjusting the repulsive force strength, \( \psi_{ij}^1, \psi_{i0}^1 \in \mathbb{R}_{>0} \) are the desired scalar separation between agents \(i\) and \(j\), and between agent \(i\) and the leader, respectively.

Let \(\mathcal{G}^{\bar{c}} = \{\mathcal{G}^1, \dots, \mathcal{G}^{\bar{C}}\}\) \(\{{\bar{c}}=1,\dots,{\bar{C}}\}\) be a finite set of possible communication topologies, where each \(\mathcal{G}^{\bar{c}} = (\mathcal{V}, \mathcal{E}^{\bar{c}}, A^{\bar{c}})\) is defined by; a set of agents \(\mathcal{V} = \{v_1,\dots,v_N\}\), an edge set \(\mathcal{E}^{\bar{c}} \subseteq \mathcal{V}\times \mathcal{V}\), and an adjacency matrix \(A^{\bar{c}}\).  A piecewise constant switching signal \(\sigma : \mathbb{R}_{\ge0}\to\{1,2,\dots,{\bar{C}}\}\) governs which topology is active at time \(t\). The signal \(\sigma\) is constant on each interval \([t_s,\,t_{s+1})\) and changes its value at \(\{t_s\}\). Hence, at any \(t\in [t_s,\,t_{s+1})\), the active topology is \(\mathcal{G}^{\sigma(t)} = (\mathcal{V},\,\mathcal{E}^{\sigma(t)},\,A^{\sigma(t)})\).
A directed edge \((v_j, v_i) \in \mathcal{E}^{\sigma(t)}\) implies that agent \(i\) can receive information from agent \(j\). The weighted adjacency matrix \(A^{\sigma(t)} = [a_{ij}^{\sigma(t)}] \in \mathbb{R}^{N \times N}\) models interaction strength under the \({\sigma(t)}\)-th topology, where:
\[
  a_{ij}^{\sigma(t)} :=
  \begin{cases}
    {z_w}_{ij}^{\sigma(t)}  & \text{if } (v_j, v_i) \in \mathcal{E}^{\sigma(t)}, \\
    0 & \text{otherwise}.
  \end{cases}
\]
Here, \( {z_w}_{ij}^{\sigma(t)} > 0 \) is the pre-defined edge weight from agent \(v_j\) to agent \(v_i\) under \(\mathcal{G}^{\sigma(t)}\).
The in-degree matrix for the \({\sigma(t)}\)-th topology is \(D^{\sigma(t)} = \text{diag}\{d_i^{\sigma(t)}\}\) with \(d_i^{\sigma(t)} = \sum_{j=1}^{N}a_{ij}^{\sigma(t)}\), and the corresponding Laplacian matrix is \(L^{\sigma(t)} = D^{\sigma(t)} - A^{\sigma(t)}\). To incorporate a leader agent \(v_0\), we define the augmented graph \(\bar{\mathcal{G}}^{\sigma(t)} = (\bar{\mathcal{V}}, \bar{\mathcal{E}}^{\sigma(t)})\), where \(\bar{\mathcal{V}} = \{v_0, v_1, \dots, v_N\}\). The leader’s influence is modeled by the diagonal matrix \(B^{\sigma(t)} = \text{diag}\{b_{i0}^{{\sigma(t)}}\}\) with \( {z_w}_{i0}^{\sigma(t)} > 0 \) the pre-defined edge weight from agent \(v_0\) to agent \(v_i\), where:
\[
  b_{i0}^{{\sigma(t)}} :=
  \begin{cases}
    {z_w}_{i0}^{\sigma(t)}  & \text{if } (v_0, v_i) \in \bar{\mathcal{E}}^{\sigma(t)}, \\
    0 & \text{otherwise}.
  \end{cases}
\]
In the time-triggered switching communication topology framework, these communication links change only at a sequence of predetermined trigger instants
\(t_s\), with \(t_0=0\) and \(t_{s+1}-t_s>0\).  Between triggers the graph remains constant.
If the inter-switch intervals satisfy a minimal dwell time 
\(\tau_d>0\), i.e.\ \(t_{s+1}-t_s\ge\tau_d\), then \(\sigma\) is said to obey an average-dwell-time constraint. Thus, at any given time, the effective communication and coordination structure is defined by \(\mathcal{G}^{\sigma(t)}\) (correspondingly \(A^{\sigma(t)}, D^{\sigma(t)}, L^{\sigma(t)}, B^{\sigma(t)}\)). 

In order to ensure that no clusters of agents are isolated from the leader, we impose the following assumption.
\begin{assumption} \label{assp1}
Under each communication topology $\mathcal{G}^\sigma$, the augmented graph \( \bar{\mathcal{G}}^{\sigma} \) contains a spanning tree with the leader as the root (the leader is (directly or indirectly) connected to every agent). In other words, whenever $(v_i, v_0) \notin \bar{\mathcal{E}}^{\sigma}$ under the
active topology, there exists a sequence of nonzero elements of \(A^{\sigma}\) of the form ${a}^{\sigma}_{i i_2}, {a}^{\sigma}_{i_2 i_3}, \cdots, {a}^{\sigma}_{i_{l-1} i^{\prime}}$, for some $(i^{\prime},0) \in \bar{\mathcal{E}}^{\sigma}$, with the sequence length $l$ being a finite integer (An agent might not communicate to the leader directly, but there is a finite hop chain of neighbors that relays leader info to it in the current graph).
\end{assumption}

In the instantaneous-connectivity setting of \citep{Tanner2007}, the communication graph is (undirected) connected at every instant, which enables arbitrarily fast switching (no dwell-time constraint) but imposes a strong per-instant topological requirement. In contrast, our framework requires leader-rooted connectivity at every instant—i.e., under Assumption~\ref{assp1} each active topology contains a spanning tree with the leader as root, so every follower is reachable (possibly through neighbors); this condition applies to both directed and undirected graphs. In addition, we enforce an average-dwell-time bound Eqn.~\eqref{numberswitches0} that limits the switching rate, ensuring sufficient time for information propagation and Lyapunov decrease. Thus, both approaches guarantee a per-instant connectivity property, but of different flavors and trade-offs: \citep{Tanner2007} uses all-time undirected connectivity with no timing limits, whereas we use leader-rooted connectivity plus an explicit switching-speed constraint—better capturing leader–follower networks subject to intermittent links or mobility-induced disconnections.

Due to the restrictive nature of the dwell time \(\tau_a\) \citep[Section 3.2.2]{Liberzon2003}, we employ average dwell time, \( \tau_a^* \) derived in Section~\ref{sec:tau_av}, that reflects the topology connectivity quality.
Let $N^{\sigma}(t_0, t)$ denote the number of discontinuities of a switching signal $\sigma$ on an interval $(t_0, t)$. We say that $\sigma$  has average dwell time \( \tau_a^* \) \citep{Hespanha} and \citep[Section 3.2.2]{Liberzon2003} if there exist two positive numbers $N_0$ and $\tau_a$ such that 
\begin{equation}\label{numberswitches0}
N^{\sigma}(t_0, t) \leq N_0 + \frac{t - t_0}{\tau_a}.
\end{equation}
Assumption \ref{assp1} guarantees that at every instant the leader’s state is reachable by every agent via a finite neighbor chain (leader-rooted spanning tree). The average-dwell-time condition Eqn.~\eqref{numberswitches0} then restricts how fast we switch among such leader-connected graphs, ensuring sufficient time in each mode for information propagation and Lyapunov decrease.

Throughout the paper, we assume that the following assumptions hold.
\begin{assumption} \label{assp}
The following conditions hold:
\begin{enumerate} 
    \item The initial states of all follower agents and the leader agent are bounded, i.e., \(||x_i (t_0)|| \le X_{\mathrm{init}}\) for all \(i \in \mathcal{N}\), and \(||x_0  (t_0)|| \le X_{0, {\mathrm{init}}}\), where \( x_i(t_0) \in \mathbb{R}^{n \cdot {\pd}}. \)

    \item There exist continuous functions \( f_{\mathrm{init}}(\cdot) \) and \( f_{0, {\mathrm{init}}}(\cdot) \) such that \( |f_i(x)| \leq |f_{{\mathrm{init}}}(x)| \) for all \(i \in \mathcal{N}\), and \( |f_0(x_0,t)| \leq |f_{0, {\mathrm{init}}}(x_0)| \) for all \(x\) within the compact sets \(\varsigma_f = \{x \mid \|x\| \leq X_{\mathrm{init}}\}\) and \(\varsigma_0 = \{x_0 \mid \|x_0\| \leq X_{0, {\mathrm{init}}}\}\). These bounds hold regardless of which topology is active at any given time.

    \item The unknown disturbances \(w_{i}(t)\) affecting each agent are bounded, i.e., \(||w(t)|| \le w_{\mathrm{init}}\) for all \(t \in [t_0,\infty)\), where \( w(t) = [ w_1(t); w_2(t); \dots; w_N(t) ] \in \mathbb{R}^{N \cdot {\pd}} \). This holds irrespective of the switching events and the currently active topology.

\end{enumerate}
\end{assumption}

To ensure that agents within a specified proximity can exchange information under the active topology, we impose the following assumption. Let $\Psi \in \mathbb{R}_{>0}$ denote a proximity threshold.

\begin{assumption}
Under any active topology, if $\| x_i^1 - x_j^1 \| \leq \Psi$, then ${z_w}_{ij}^{\sigma(t)} \neq 0$.
\end{assumption}

This assumption enforces that any pair of agents within $\Psi$-distance are connected by the communication graph. 

\section{Methodology} \label{Methodology}
\subsection{Neural Network Learning Problem}
To account for uncertainties in the dynamics under time-triggered switching communication topologies, we use a distributed two-layer linear-in-parameters (LIP) neural network \citep{yesildirak1995neural}, enabling each agent to locally update its model as the communication graph changes. This framework allows each agent to approximate its own non-linear dynamics and those of its neighbors while sharing their states as determined by the currently active topology.
Accordingly, we model the functions as:
\begin{align} 
\label{orige}
f_{i}(x_{i}) &= \theta_{i}^{\top}\phi_{i}(x_{i}) + \varepsilon_i, \\ 
\label{orige1}
f_{0}(x_{0},t) &= \theta_{0}^{\top}\phi_{0}(x_{0},t) + \varepsilon_{0}, \\
\label{orige2}
w_{i}(t) &= \theta_{iw}^{\top}\phi_{iw}(t) + \varepsilon_{iw}.
\end{align}
where \( f_{i}(x_{i}) \in \mathbb{R}^{{\pd}} \) is the non-linear function in agent \( i \)'s dynamics, \( \phi_{i}(x_{i}) \in \mathbb{R}^{q_i} \), \( \phi_{0}(x_{0},t)  \in \mathbb{R}^{q_0} \), and \( \phi_{iw}(t)   \in \mathbb{R}^{q_{iw}} 
\) are fixed basis functions, with the corresponding weight vectors \( \theta_{i} \in \mathbb{R}^{{\pd} \cdot q_i} \), \( \theta_{0} \in \mathbb{R}^{{\pd} \cdot q_0} \), and \( \theta_{iw} \in \mathbb{R}^{{\pd} \cdot q_{iw}}\), updated in real-time as new data is received from neighbors defined by the active topology. The leader's non-linear function is \( f_{0}(x_{0}, t) \in \mathbb{R}^{{\pd}} \), and \( w_{i}(t) \in \mathbb{R}^{{\pd}} \) represents the unknown disturbance affecting agent \( i \). The approximation errors are \( \varepsilon_i  \in \mathbb{R}^{{\pd}} \), \( \varepsilon_0  \in \mathbb{R}^{{\pd}}\), and \( \varepsilon_{iw}  \in \mathbb{R}^{{\pd}}\), \( q_i \) is the number of basis functions for agent \( i \), \( q_{0} \) is the number of basis functions for the leader agent, and \( q_{iw} \) is the number of basis functions for modeling disturbance.

The neural network approximations for \( f_{i}(x_{i}) \), \( f_{0}(x_{0},t) \), and \( w_{i}(t) \) are therefore
\(
\hat{f}_{i}(x_{i}) = \hat{\theta}_{i}^{\top}\phi_{i}(x_{i}), \quad \hat{f}_{0}(x_{0},t) = \hat{\theta}_{0}^{\top}\phi_{0}(x_{0},t), \quad \hat{w}_{i}(t) = \hat{\theta}_{iw}^{\top}\phi_{iw}(t)
\), where the NN weights \( \hat{\theta}_{i} \), \( \hat{\theta}_{0} \), and \(\hat{\theta}_{iw}\) are computed locally. 
For brevity of notation, the basis functions and errors are denoted by
\[
\theta = \text{diag}(\theta_1, \ldots, \theta_N) \in \mathbb{R}^{N \cdot {\pd} \times q_t},\] \[
\phi(x) = [\phi_1^{\top}(x_1), \ldots, \phi_N^{\top}(x_N)]^{\top} \in \mathbb{R}^{q_t},\] \[
\varepsilon = [\varepsilon_1, \ldots, \varepsilon_N]^{\top} \in \mathbb{R}^{N \cdot {\pd}}.\] Also, \( q_t \) is the total number of basis functions for all agents, \( q_{w} = \sum_{i=1}^{N} q_{iw} \), and \( q_t = \sum_{i=1}^{N} q_i \).

The global non-linearities are expressed as:
\begin{align}\label{fxes}
     f(x) &= {\theta}^{\top} \phi(x) + \varepsilon 
\\
\label{f0es}
     f_0(x,t) &= {\theta}_0^{\top} \phi_0(x,t) + \varepsilon_0
\\
\label{fwes}
     w(t) &= {\theta}_w^{\top} \phi_w(t) + \varepsilon_w 
\end{align}
with approximations:
\begin{align}\label{fxesh}
      \hat{f}(x) &= \hat{\theta}^{\top} \phi(x)  
\\
\label{f0esh}
      \hat{f}_0(x,t) &= \hat{\theta}_0^{\top} \phi_0(x,t) 
\\
\label{fwesh}
      \hat{w}(t) &= \hat{\theta}_w^{\top} \phi_w(t)
\end{align}
where $ f(x) \in \mathbb{R}^{N \cdot {\pd}} $ is the concatenated non-linear function for all agents, \( \theta_{w} = \operatorname{diag}(\theta_{1w}, \ldots, \theta_{Nw}) \in \mathbb{R}^{N \cdot {\pd} \times q_{w}} \), \( \phi_w(t) = [\phi_{1w}(t), \phi_{2w}(t), \ldots \phi_{Nw}(t)] \in \mathbb{R}^{q_{w}} \), \( \varepsilon_{w} = [\varepsilon_{1w}, \varepsilon_{2w}, \ldots, \varepsilon_{Nw}] \in \mathbb{R}^{N \cdot {\pd}} \). 
The NN weight errors are:
\begin{align} \label{eq:nn_1}
\tilde{\theta} &= \theta - \hat{\theta}, \\
\label{eq:nn_2}  
\tilde{\theta}_{0} &= \theta_{0} - \hat{\theta},_{0}\\
\label{eq:nn_3}  
\tilde{\theta}_{w} &= \theta_{w} - \hat{\theta}_{w}.
\end{align}

\begin{assumption} \label{assumption1}
The basis functions \( \phi_{i}(x_{i}) \), \( \phi_{0}(x_{0}, t) \), \( \phi_{iw}(t) \), the NN weights \( \theta_{i} \), \( \theta_{0} \), \( \theta_{iw} \), and the approximation errors \( \varepsilon \), \( \varepsilon_{0} \), \( \varepsilon_{w} \) are bounded by finite constants (not required to be known a priori)., independent of the currently active topology. 
\end{assumption}

Let \( \phi_{in} = \max_{x_i \in \varsigma} \|\phi_i(x_i)\| \), \( \phi_{0n} = \max_{x_0 \in \varsigma, t \geq 0} \|\phi_0(x_0, t)\| \), and \( \phi_{iwn} = \max_{t \geq 0} \|\phi_{iw}(t)\| \). Based on 
Assumption~\ref{assumption1},
there exist positive numbers \( \Phi_n \), \( \Theta_n \), and \( \varepsilon_n \); \( \Phi_{n0} \), \( \Theta_{n0} \), and \( \varepsilon_{n0} \); \( \Phi_{nw} \), \( \Theta_{nw} \), and \( \varepsilon_{nw} \), such that:
\(
\|\phi_i(x_i)\| \leq \Phi_n, \quad \|\phi_0(x_0, t)\| \leq \Phi_{n0}, \quad \|\phi_{iw}(t)\| \leq \Phi_{nw},
\)
\(
\|\theta_i\|_G \leq \Theta_n, \quad \|\theta_0\|_G \leq \Theta_{n0}, \quad \|\theta_{iw}\|_G \leq \Theta_{nw},
\)
\(
\|\varepsilon\| \leq \varepsilon_n, \quad \|\varepsilon_0\| \leq \varepsilon_{n0}, \quad \|\varepsilon_w\| \leq \varepsilon_{nw}.
\)

The use of a mixed dictionary of polynomial, trigonometric, and exponential basis functions is especially well-suited to adaptive control \citep{stevenbrunton2016}, since each class brings complementary approximation strengths—polynomials excel at capturing local algebraic trends, trigonometric functions model periodic behaviors, and exponentials describe transient phenomena. As demonstrated in the SINDy framework of \citep{stevenbrunton2016}, such hybrid libraries retain universal‐approximation power while keeping the adaptation problem tractable: their simple, closed‐form derivatives lead to well‐structured, real‐time optimization routines. Compared with large‐scale neural networks, these basis functions yield compact adaptation laws with lower computational overhead, making them ideal for real-time control of systems with smooth nonlinearities and structured, time-varying disturbances~\citep{stevenbrunton2016}.

\subsection{Local Synchronization Error}
With time-triggered switching communication topologies, at any given time \( t \), the active adjacency matrix \( A^{\sigma(t)} \), Laplacian \( L^{\sigma(t)} = D^{\sigma(t)} - A^{\sigma(t)} \), and leader interaction matrix \( B^{\sigma(t)} \) define the communication structure. Given the limited information available to the \( i^\text{th} \) agent under the currently active topology \(\mathcal{G}^{\sigma(t)}\), we define its \( k^\text{th} \) order weighted synchronization error as:
\begin{multline}{\label{errordyn}}
e_i^{k,\sigma(t)} = - \nu_1 \sum_{j=1}^{N} a_{ij}^{\sigma(t)} \bigg[(x_{i}^{k} - \psi_{i}^{k}) - (x_{j}^{k} - \psi_{j}^{k})\bigg] \\
- \nu_2 b_{i0}^{\sigma(t)} \bigg[(x_{i}^{k} - \psi_{i}^{k}) - (x_{0}^{k} - \psi_{0}^{k})\bigg],
\end{multline}
where \(\nu_1, \nu_2 > 0\) are scalar gains, \(a_{ij}^{\sigma(t)}\) are the entries of the active adjacency matrix \( A^{\sigma(t)} \), and \(b_{i0}^{\sigma(t)}\) are the entries of the leader interaction matrix \( B^{\sigma(t)} \). For notational convenience, we employ $e_i^{k,\sigma(t)}$ and $e_i^{k}$ interchangeably throughout the manuscript to represent the same mathematical entity. The scalar gains \(\nu_1, \nu_2 \) can vary between topologies, however, in this paper, we consider them constant across topologies.
Using the simplifying notations:
\(\bar{x}^k_i := x_{i}^{k} - \psi_{i}^{k}\),
\(\bar{x}^k_j :=  x_{j}^{k} - \psi_{j}^{k}\) and
\(\bar{x}^k_0 :=  x_{0}^{k} - \psi_{0}^{k}\),
we rewrite Eqn. \eqref{errordyn} as:
\begin{align}{\label{errordynu}}
e_i^{k,\sigma(t)} = - \nu_{1} \sum_{j=1}^{N} a_{ij}^{\sigma(t)} (\bar{x}^k_i - \bar{x}^k_j)
- \nu_{2} b_{i0}^{\sigma(t)} (\bar{x}^k_i - \bar{x}^k_0).
\end{align}
Here, the term \(- \nu_{1} \sum_{j=1}^{N} a_{ij}^{\sigma(t)} (\bar{x}^k_i - \bar{x}^k_j)\)  ensures consensus among agents by penalizing the difference between the state of agent \(i\) and its neighbors, while \(- \nu_{2} b_{i0}^{\sigma(t)} (\bar{x}^k_i - \bar{x}^k_0)\) couples the agent's state to the leader state by penalizing their difference. Defining,
\(
{e}^{k,\sigma(t)} = [{e}^{k,\sigma(t)}_{1}, {e}^{k,\sigma(t)}_{2}, {e}^{k,\sigma(t)}_{3}, \dots, {e}^{k,\sigma(t)}_{N}]^\top \in \mathbb{R}^{N \cdot {\pd}},
\)
\(
\bar{\underline{x}}_{0}^k = [\bar{x}_{0}^k, \dots, \bar{x}_{0}^k]^\top\in \mathbb{R}^{N \cdot {\pd}},
\)
\(
\bar{x}^{k} = [\bar{x}_{1}^{k},\dots, \bar{x}_{N}^{k}]^{\top}
\in \mathbb{R}^{N \cdot {\pd}}\),
we reformulate Eqn.~\eqref{errordynu} in global form as:
\begin{equation}
e^{k,\sigma(t)} = -(\nu_1 L^{\sigma(t)} + \nu_2 B^{\sigma(t)})(\bar{x}^k - \bar{\underline{x}}_{0}^k),
\end{equation}
highlighting that the synchronization error depends on the currently active topology.

\begin{lemma} \label{rem:lbnon}
Under~~Assumption~~\ref{assp1},~~the~~matrices {\(\nu_1 L^{\sigma(t)} + \nu_2 B^{\sigma(t)}\)} are nonsingular for all topologies \(\sigma(t)\).
\end{lemma}

\begin{proof} 
Let \(J\) denote the set of strictly diagonally dominant nodes of \(\bm{\pounds}^{\sigma(t)}:=\nu_1 L^{\sigma(t)} + \nu_2 B^{\sigma(t)}\), i.e,
\(
J = \{ i : |\pounds_{ii}^{\sigma(t)}| > \sum_{j \neq i} |\pounds_{ij}^{\sigma(t)}| \}.
\) 
Similar to the proof under a fixed topology, if \(J = \mathcal{N}\), i.e., the matrix \(\bm{\pounds}^{\sigma(t)}\) is strictly diagonally dominant, then  \(\det(\bm{\pounds}^{\sigma(t)}) \neq 0\) is directly obtained by the application of the Gershgorin circle theorem \citep{qu2009cooperative}. 
For the case where \(J \subset \mathcal{N} \) is a strict subset of agents, we note from $\bm{\pounds}^{\sigma(t)}:=\nu_1 L^{\sigma(t)} + \nu_2 B^{\sigma(t)}$ and $L^{\sigma(t)}:=D^{\sigma(t)}-A^{\sigma(t)}$ that \mbox{$\bm{\pounds}^{\sigma(t)} = \nu_1 D^{\sigma(t)} - \nu_1 A^{\sigma(t)} +\nu_2 B^{\sigma(t)}$}. Since $B^{\sigma(t)} := \text{diag}\{b_{i}^{0}\}$ and \mbox{$D := \text{diag}\{d_i^{\sigma(t)}\}$} are diagonal matrices, the on-diagonal elements of $\bm{\pounds}^{\sigma(t)}$ are $\bm{\pounds}^{\sigma(t)}_{ii} = -\nu_1 d_{ii} -\nu_2 b_{i}^{0}$ and the off-diagonal elements are $\bm{\pounds}^{\sigma(t)}_{ij} = -\nu_1 a_{ij}$, $i \neq j$. 
These relations, together with Assumption~\ref{assp1}, yields that: (i) \(J \neq \emptyset\) and (ii) for each agent \(i \notin J\), there exists a sequence of nonzero elements of $\bm{\pounds}^{\sigma(t)}$ in the form ${\pounds}_{i_1 i_2}, {\pounds}_{i_2 i_3}, \ldots, {\pounds}_{i_{s-1} i_s}$, for some $i_s \in J$, which directly correspond to the sequence of non-zero elements $ a_{i_1 i_2}, a_{i_2 i_3}, \ldots, a_{i_{s-1} i_s}$ in Assumption~\ref{assp1} (applied to each active topology \(\mathcal{G}^{\sigma(t)}\)). Thus, by invoking \cite[Theorem]{shivakumar1974sufficient}, we conclude that \(\bm{\pounds}^{\sigma(t)}\) is nonsingular under any active topology.
\hfill $\blacksquare$
\end{proof}
\subsection{Local Error Dynamics}
Under time-triggered switching communication topologies, the active Laplacian and leader-interaction matrices at time \( t \) are \( L^{\sigma(t)} \) and \( B^{\sigma(t)} \). For each time $t$, the differentiation of the synchronization error defined in Eqn.~\eqref{errordynu} yields:
\begin{align}
\label{plantas}
\dot{e}^{k,\sigma(t)} &= e^{k+1,\sigma(t)}, \quad k=1,\ldots,n-1\notag \\
\dot{e}^{k,\sigma(t)} &= -(\nu_{1} L^{\sigma(t)} + \nu_{2} B^{\sigma(t)})(\dot{\bar{x}}^n - \dot{\bar{\underline{x}}}_0^n), \quad k=n
\end{align}
which, in the expanded form, is written as:
\begin{align} \label{explantas}
\dot{e}^1 &= e^2 = -(\nu_1 L^{\sigma(t)}+\nu_2 B^{\sigma(t)})(\bar{x}^2 - \bar{\underline{x}}_{0}^2) \notag \\
\dot{e}^2 &= e^3 = -(\nu_1 L^{\sigma(t)}+\nu_2 B^{\sigma(t)})(\bar{x}^3 - \bar{\underline{x}}_{0}^3) \notag \\
&\vdots \notag \\
\dot{e}^n &= -(\nu_1 L^{\sigma(t)}+\nu_2 B^{\sigma(t)})(f(\bar{x}) + u + w - f_0(\bar{x}_0,t))
\end{align}
where \( {e}^{k,\sigma(t)} \in \mathbb{R}^{N \cdot {\pd}} \), \( \bar{x}^k \in \mathbb{R}^{N \cdot {\pd}} \), and \( \bar{\underline{x}}_0^k \in \mathbb{R}^{N \cdot {\pd}} \). Here, \(f(\bar{x})\), \(f_0(\bar{x}_0,t)\), \(u\), and \(w\) represent the aggregated non-linearities, leader dynamics, control input, and disturbances, respectively, and evolve under the currently active communication topology defined by \(\sigma(t)\).

\subsection{Weighted Stability Error}
The weighted stability error \(r_i \in \mathbb{R}^{{\pd}}\) for each follower agent~\(i\) is defined as:
\begin{align}
r_i =& ~\lambda_{1}e_{i }^1 + \lambda_{2}e_{i }^2 + \dots + \lambda_{n-1}e_{i}^{n-1} + e_{i}^n   \notag \\
=& - \nu_1 \sum_{k=1}^{n} \sum_{j=1}^{N} \lambda_{k} a_{ij}^{\sigma(t)} \bigg[(x_{i}^{k} - \psi_{i}^{k}) - (x_{j}^{k} - \psi_{j}^{k})\bigg] \notag \\
&- \nu_2 \sum_{k=1}^{n} \lambda_{k}  b_{i0}^{\sigma(t)} \bigg[(x_{i}^{k} - \psi_{i}^{k}) - (x_{0}^{k} - \psi_{0}^{k})\bigg]
\end{align}
with $\lambda_n = 1$, and
where \(\lambda_j\) (for \( j = 1, \ldots, n-1 \)) the design parameters are chosen such that the characteristic polynomial:
\begin{equation} \label{eq:hurwitz_poly}
\mathfrak{s}^{n-1} + \lambda_{n-1} \mathfrak{s}^{n-2} + \cdots + \lambda_1
\end{equation}
is Hurwitz. This ensures that all roots have negative real parts, leading to the stability of the associated linear system. The selection of \(\lambda_j\) can be made by setting
$\mathfrak{s}^{n-1} + \lambda_{n-1} \mathfrak{s}^{n-2} + \cdots + \lambda_1 = \prod_{j=1}^{n-1} (\mathfrak{s} - \xi_j)
$, and selecting \(\xi_j\)'s to be positive real numbers. The global form of the representation of the weighted stability error is:
\begin{equation} \label{slidem}
r = \lambda_{1}e^{1} + \lambda_{2}e^{2} + \dots + \lambda_{n-1}e^{n-1} + e^{n} \in \mathbb{R}^{N\cdot p  \times 1}.
\end{equation}
Similarly, we define
\begin{align}
 \rho_i &= \lambda_{1}e_{i }^2 + \lambda_{2}e_{i}^3 + \dots + \lambda_{n-1}e_{i}^{n}   
 \notag \\
&= - \nu_1 \sum_{k=2}^{n} \sum_{j=1}^{N} \lambda_{k-1} a_{ij}^{\sigma(t)} \bigg[(x_{i}^{k} - \psi_{i}^{k}) - (x_{j}^{k} - \psi_{j}^{k})\bigg] \notag \\
&- \nu_2 \sum_{k=2}^{n} \lambda_{k-1}  b_{i0}^{\sigma(t)} \bigg[(x_{i}^{k} - \psi_{i}^{k}) - (x_{0}^{k} - \psi_{0}^{k})\bigg]
\end{align}
as well as  
\begin{equation} \label{eq:rhoDefinition}
    \rho = \lambda_{1}{e}^{2} + \lambda_{2}{e}^{3} + \lambda_{3}{e}^{4} + \dots + \lambda_{n-1}{e}^{n}.
\end{equation}

By differentiating Eqn.~\eqref{slidem} with respect to time, the weighted stability error dynamics under the active topology \(\sigma(t)\) is written as:
\begin{multline}
\label{slidemd}
\hspace{-9pt}\dot{r} = \lambda_{1}\dot{e}^{1} + \lambda_{2}\dot{e}^{2} + \dots + \lambda_{n-1}\dot{e}^{n-1} + \dot{e}^{n} \\
 = \lambda_{1}{e}^{2} + \lambda_{2}{e}^{3} + \dots + \lambda_{n-1}{e}^{n} 
 \\
 - (\nu_{1}L^{\sigma(t)} + \nu_{2}B^{\sigma(t)})(\dot{\bar{x}}^n - \dot{\bar{x}}_0^n) \notag
\end{multline}
recalling from Eqn.~\eqref{plantas} that \(\dot{e}^{n} = - (\nu_{1}L^{\sigma(t)} + \nu_{2}B^{\sigma(t)})(\dot{\bar{x}}^n - \dot{\bar{x}}_0^n).\) Hence, \(\dot{r}\) can be written as 
\begin{equation} \label{errdyn}
\dot{r} = \rho - (\nu_{1}L^{\sigma(t)} + \nu_{2}B^{\sigma(t)})(f(x) + u + w - f_0(x_0,t)),
\end{equation}
where 
\begin{equation}\label{eq:desparam}
\bar{\lambda} = [\lambda_1, \lambda_2,\ldots,\lambda_{n-1}]^{\top} \end{equation}  

With the definition of \(
E_1 = [e^1,e^2,\ldots,e^{n-1}] \in \mathbb{R}^{N\cdot p  \times (n-1)},
\) and \( E_2=[{e^2}, {e^3},\ldots,{e^n}] \in \mathbb{R}^{N\cdot p  \times (n-1)} \), we first rewrite \eqref{eq:rhoDefinition} as $\rho = E_{2}\bar{\lambda}$ and, further, we can write the matrix relation: 
\begin{equation}\label{eq:E2_def}
E_2 = E_1 \bigtriangleup^{\top} + \, r \, l =  \dot{E}_1 = [e^2,e^3,\ldots,e^n] \in \mathbb{R}^{1 \times N\cdot p \cdot(n-1)}
\end{equation}
with
\mbox{\(
l = [0,0,\ldots,0,1] \in \mathbb{R}^{1 \times (n-1)}
\)}, and 
\[
\bigtriangleup = \begin{bmatrix}
0 & 1 & 0 & \ldots & 0 \\
0 & 0 & 1 & \ldots & 0 \\
\vdots & \vdots & \vdots & \ddots & \vdots \\
0 & 0 & 0 & \ldots & 1 \\
-\lambda_1 & -\lambda_2 & -\lambda_3 & \ldots & -\lambda_{n-1} \\
\end{bmatrix} \in \mathbb{R}^{(n-1) \times (n-1)}.
\]

As stated in \citep{lewis2013}, since $\bigtriangleup$ is Hurwitz,  there exists a unique positive-definite matrix $P_1$ such that for any positive number \( {\beta} \):
\begin{equation} \label{huwitx}
    \bigtriangleup^{\top} P_1 + P_1 \bigtriangleup = - {\beta}I
\end{equation}
where \( I \in \mathbb{R}^{(n-1)\times(n-1)} \) is the identity matrix.
We will use $P_1$ in the Lyapunov function to establish convergence.

\begin{lemma}
If \( r_i(t) \) is ultimately bounded, then \( e_i(t) \) is ultimately bounded. 
\end{lemma}
\begin{proof}
    The result is direct consequence of \citep[Lemma 10.3]{lewis2013}. 
\end{proof}

\begin{lemma}[Graph Lyapunov Equation] \label{lem:PQpositivedefinite}
Define 
\begin{equation}
    q \equiv [q_1,\ldots,q_N]^{\top} := (\nu_{1}L^{\sigma(t)} + \nu_{2}B^{\sigma(t)})^{-1} \underline{1} ,
\end{equation}
\begin{equation} \label{pdefinite}
    P := \text{diag}\{1/q_i\}_{i \in \mathcal{N}},
\end{equation}
\begin{equation} \label{qdefinite}
    Q:=P(\nu_1L^{\sigma(t)} + \nu_2B^{\sigma(t)}) + (\nu_1L^{\sigma(t)} + \nu_2B^{\sigma(t)})^{\top} P ,
\end{equation}
where \( \underline{1} = [1,\ldots,1]^{\top} \in \mathbb{R}^N \). 
Then the matrices \( P \) and \( Q \) are positive definite.
\end{lemma}
\begin{proof}
      The proof is similar to \citep[Lemma 10.1]{lewis2013} by arguing that \(\nu_1 L^{\sigma(t)} + \nu_2 B^{\sigma(t)}\) is nonsingular by Lemma~\ref{rem:lbnon}.
      
      \hfill $\blacksquare$
\end{proof}

\ignorespacesafterend

\section{Main Result} 
\label{MainResult}
We now introduce the proposed
distributed control strategy and adaptive tuning laws in Sections~\ref{control_law_defn} and~\ref{NN_Laws_Defn}, respectively.
\subsection{Proposed distributed Control Law}\label{control_law_defn}
Under time-triggered switching communication topologies, let \( D^{\sigma(t)} \) and \( B^{\sigma(t)} \) be the in-degree and leader-interaction matrices corresponding to the active topology \(\sigma(t)\). Given, \( {c}_i \in \mathbb{R}^{{\pd} \times n{\pd}} \) a tunable gain matrix and 
\begin{equation}
    E_{i0}=\left[\begin{array}{c}
(x_{i}^{1}-\psi_{i}^{1})-(x_{0}^{1}-\psi_{0}^{1})\\
\vdots\\
(x_{i}^{k}-\psi_{i}^{k})-(x_{0}^{k}-\psi_{0}^{k})\\
\vdots\\
(x_{i}^{n}-\psi_{i}^{n})-(x_{0}^{n}-\psi_{0}^{n})
\end{array}\right]
\end{equation}
is the relative position between the follower and leader with global forms represented as \( {c} \) and \(E_0\), respectively.
To achieve a fully modular design that simultaneously handles formation tracking, disturbance rejection, consensus, leader-following, and safety, our control law for high-dimensional nonlinear agent $i$ is composed of five clear components. First, the feedforward term $\rho_i$ encodes the desired leader-follower offset $\psi_i^{k}$, ensuring each agent follows the leader ($\psi_i^{1}$ is the relative displacement). Second, the adaptive neural-network terms
$
-\hat\theta_i^\top\phi_i \;-\;\hat\theta_{iw}^\top\phi_{iw}\;+\;\hat\theta_0^\top\phi_0
$
cancel unknown agent dynamics and external disturbances via online weight adaptation. Third, the synchronization error $r_i$ drives all inter-agent and agent-leader deviations—after subtracting offset thresholds—to zero, enforcing cohesion. Fourth, a proportional leader-tracking correction
$-c_i\,E_{i0}^\top$
guarantees exponential convergence to the leader’s trajectory through the communication graph. Finally, the repulsive safety forces
$\Gamma^0_{\mathfrak{c}} \sum_{{\mathfrak{c}}=1}^{\xi} m_{i\mathfrak{c}},\Gamma^1_{ij} \sum_{j=1}^{N} m_{ij}, \Gamma^2_{i0} \sum_{j=1}^{N} m_{i0}$
provide collision- and obstacle-avoidance by pushing agents away from nearby hazards. 

Each term directly maps to one of our core objectives, resulting in a tunable control architecture.
The proposed distributed control law is presented for each agent~\(i\) as:

\begin{equation} \label{eq:ProposedControlLaw}
\begin{aligned}
u_i &= \frac{\rho_i}{d_i^{\sigma(t)}+b_{i0}^{\sigma(t)}} - \hat{\theta}_i^{\top}{\phi}_i - \hat{\theta}_{iw}^{\top}{\phi_{iw}} +\hat{\theta}_0^{\top}{\phi_0} + r_i \\
&\quad - {c}_i E_{i0}^\top -  \sum_{{\mathfrak{c}}=1}^{\Xi} \Gamma^0_{\mathfrak{c}} m_{i\mathfrak{c}} - \Gamma^1_{ij} m_{ij} -  \Gamma^2_{i0} \sum_{j=1}^{N} m_{i0}
\end{aligned}
\end{equation}
and written in expanded form

\begin{strip}
\begin{multline}\label{eq:ProposedControlLaw1}
u_{i}=\frac{1}{d_{i}^{\sigma(t)}+b_{i0}^{\sigma(t)}}\Bigg(-\nu_{1}\sum_{k=2}^{n}\sum_{j=1}^{N}\lambda_{k-1}a_{ij}^{\sigma(t)}\bigg[(x_{i}^{k}-\psi_{i}^{k})-(x_{j}^{k}-\psi_{j}^{k})\bigg]-\nu_{2}\sum_{k=2}^{n}\lambda_{k-1}b_{i0}^{\sigma(t)}\bigg[(x_{i}^{k}-\psi_{i}^{k})-(x_{0}^{k}-\psi_{0}^{k})\bigg]\Bigg)
\\
-\hat{\theta}_{i}^{\top}\phi_{i}(x_{i})-\hat{\theta}_{iw}^{\top}\phi_{iw}(t)+\hat{\theta}_{0}^{\top}\phi_{0}(x_{0},t)-\nu_{1}\sum_{k=1}^{n}\sum_{j=1}^{N}\lambda_{k}a_{ij}^{\sigma(t)}\bigg[(x_{i}^{k}-\psi_{i}^{k})-(x_{j}^{k}-\psi_{j}^{k})\bigg]
\\
-\nu_{2}\sum_{k=1}^{n}\lambda_{k}b_{i0}^{\sigma(t)}\bigg[(x_{i}^{k}-\psi_{i}^{k})-(x_{0}^{k}-\psi_{0}^{k})\bigg]-\sum_{k=1}^{n}c_{i}^{k}\bigg[(x_{i}^{k}-\psi_{i}^{k})-(x_{0}^{k}-\psi_{0}^{k})\bigg]
\\
-\sum_{\mathfrak{c}=1}^{\Xi}\Gamma_{\mathfrak{c}}^{0}\,m_{i\mathfrak{c}}(x_i,O_{\mathfrak{c}})-\sum_{j=1}^{N}\Gamma_{ij}^{1}\,m_{ij}(x_i,x_j)-\Gamma_{i0}^{2}\,m_{i0}(x_i,x_0).
\end{multline}
\end{strip}
To facilitate the discussion, Eqn.~\eqref{eq:ProposedControlLaw} is decomposed as
\begin{align}
    \hspace{-3pt}u_i^d &= \frac{\rho}{d_i^{\sigma(t)}+b_{i0}^{\sigma(t)}} - \hat{\theta}_i^{\top}{\phi}_i - \hat{\theta}_{iw}^{\top}{\phi_{iw}} +\hat{\theta}_0^{\top}{\phi_0} \notag \\ &+ r_i - {c}_i E_{i0}^\top,  \\
    \hspace{-3pt}u_i^0 &= \Gamma^0_{\mathfrak{c}} \sum_{{\mathfrak{c}}=1}^{\xi} m_{i\mathfrak{c}}, \\
    \hspace{-3pt}u_i^c &= \Gamma^1_{ij} \sum_{j=1}^{N} m_{ij} +  \Gamma^2_{i0} \sum_{j=1}^{N} m_{i0}.
\end{align}
This yields the global form of the control law \( U \) as:
\begin{equation}\label{controllaw}
    U = U^d - U^0 - U^c
\end{equation}
where:
\begin{align}
    \hspace{-3pt} U^d &= {(D^{\sigma(t)}+B^{\sigma(t)})}^{-1}{\rho} - \hat{\theta}^{\top}{\phi} - \hat{\theta}_{w}^{\top}{\phi_{w}} + \hat{\theta}_0^{\top}{\phi_0} \notag \\ &+ r - {c} E_0^\top,  \\
    \hspace{-3pt} U^0 &= \Gamma^0 M^0, \\
    \hspace{-3pt} U^c &= \Gamma^1 M^c_I + \Gamma^2 M^c_0.
\end{align}
Gain matrices $\Gamma^0,\Gamma^1,\Gamma^2$ multiply the obstacle‐avoidance term $M^0$ and the two collision‐avoidance terms $M^c_I,\;M^c_0$.  Each $\Gamma^j$ is a constant, positive-definite matrix chosen by the designer to set the relative priority and shape of those repulsive contributions. The overall control law is:
\begin{align}\label{overcntllaw}
    U &= {(D^{\sigma(t)}+B^{\sigma(t)})}^{-1}{\rho} - \hat{\theta}^{\top}{\phi} - \hat{\theta}_{w}^{\top}{\phi_{w}} +\hat{\theta}_0^{\top}{\phi_0} \notag \\ &+ r - {c} E_0^\top - \Gamma^0 M^0 - \Gamma^1 M^c_I - \Gamma^2 M^c_0
\end{align}
Here,  \( U = [u_1^\top, u_2^\top, \dots, u_N^\top]^\top \in \mathbb{R}^{N \cdot {\pd}} \) is the control input vector for all agents with  \( u_i, u_i^d, u_i^c, u_i^0 \in \mathbb{R}^{{\pd}} \), and \((D^{\sigma(t)}+B^{\sigma(t)})^{-1}\) is the inverse of the diagonal matrix \( \text{diag}[({d}^{\sigma(t)}_{1}+b^{\sigma(t)}_{10})I_{{\pd}}, ({d}^{\sigma(t)}_{2}+b^{\sigma(t)}_{20})I_{{\pd}}, \dots, ({d}^{\sigma(t)}_{N}+b^{\sigma(t)}_{N0})I_{{\pd}}]  \in \mathbb{R}^{(N \cdot {\pd} ) \times (N \cdot {\pd} )}\). \( I_{{\pd}} \) is the identity matrix of size \( {\pd} \times {\pd} \). The global neural network weight matrices \( \hat{\theta}, \hat{\theta}_{w}, \hat{\theta}_0 \) correspond to the agents' dynamics, disturbances, and leader's dynamics, while the global basis function matrices \( \phi, \phi_w, \phi_0 \) relate to these respective dynamics. The control law \(U\) thus adapts to the currently active topology \(\sigma(t)\), ensuring performance and stability despite the evolving communication structure.

\subsection{Proposed NN Local Tuning Laws}\label{NN_Laws_Defn} 
At any time \( t \), the active topology \(\sigma(t)\) determines the in-degree matrix \(D^{\sigma(t)}\) and the leader-interaction matrix \(B^{\sigma(t)}\).
For each agent \( i \), the NN parameter adaptive tuning laws are proposed as:
\begin{equation} \label{eq:ThetaItuning}
\dot{\hat{\theta}}_i = - G_i \big[\phi_i r_i p_i(d_i^{\sigma(t)} + b_{i0}^{\sigma(t)}) + \kappa_i\hat{\theta}_i \big],
\end{equation}
\begin{equation} \label{eq:Theta0tuning}
\dot{\hat{\theta}}_0 = G_{i0}\big[\phi_0 r_i p_i(d_i^{\sigma(t)} + b_{i0}^{\sigma(t)}) - \kappa_0\hat{\theta}_0 \big],
\end{equation}
\begin{equation} \label{eq:ThetaIWtuning}
\dot{\hat{\theta}}_{iw} = - G_{iw} \big[\phi_{iw} r_i p_i(d_i^{\sigma(t)} + b_{i0}^{\sigma(t)}) + \kappa_{iw}\hat{\theta}_{iw} \big].
\end{equation}
where, $ \kappa_i, \kappa_0, \kappa_{iw} \in \mathbb{R}_{>0} $ are positive scalar tuning gains. Here $G_i\in\mathbb R^{q_i\times q_i},\;G_{i0}\in\mathbb R^{q_0\times q_0},$ and $G_{iw}\in\mathbb R^{q_{iw}\times q_{iw}}$ are constant, symmetric positive–definite matrices that serve as the adaptation‐rate (gain) matrices for the agent–dynamics, leader–dynamics and disturbance neural networks respectively.  Their magnitudes can be tuned to trade off convergence speed versus robustness to noise or unmodelled dynamics. In the global context, this yields the parameter estimate dynamics as:
\begin{align} \label{adap1}
\dot{\hat{\theta}} &= - G \big[\phi r^{\top} P(D^{\sigma(t)} + B^{\sigma(t)}) + \kappa\hat{\theta} \big],
\\ \label{adapt2}
\dot{\hat{\theta}}_0 &= G^*_0\big[\phi_0 r^{\top} P(D^{\sigma(t)} + B^{\sigma(t)}) - \kappa_0\hat{\theta}_0 \big],
\\  \label{adapt3}
\dot{\hat{\theta}}_{w} &= - G_{w} \big[\phi_{w} r^{\top} P(D^{\sigma(t)} + B^{\sigma(t)}) + \kappa_{w}\hat{\theta}_{w} \big].
\end{align}
where 
\begin{align*}
    G &= \text{diag}(G_1,\ldots,G_N) \in \mathbb{R}^{q_t \times q_t},
    \\
    G^*_0 &= \text{diag}(G_0,\ldots,G_0) \in \mathbb{R}^{q_t \times q_t},
    \\
    G_w &= \text{diag}(G_{1w},\ldots,G_{Nw}) \in \mathbb{R}^{q_t \times q_t}.
\end{align*}

The matrix \(P\) is defined as in Eqn.~\eqref{pdefinite}, and \(\kappa\), \(\kappa_{0}\), and \(\kappa_{w}\) are vectors of the tuning gains $ \kappa_i, \kappa_0, \kappa_{iw}$ respectively.

\subsection{Leader-Follower Formation Consensus} \label{leader_follower_cons}

\begin{theorem} \label{thm:unified} 
For the distributed multi-agent with the follower dynamics Eqn.~\eqref{planta} and leader dynamics Eqn.~\eqref{plantref} under assumptions \ref{assp}, \ref{assp1}, and \ref{assumption1}, the NN tuning laws Eqn.~\eqref{eq:Theta0tuning}, Eqn.~\eqref{eq:ThetaItuning}, Eqn.~\eqref{eq:ThetaIWtuning}, together with the control law Eqn.~\eqref{eq:ProposedControlLaw}, and for a network with the average dwell time Eqn.~\eqref{eq:dwelltime}, this
results in synchronization,
collision/obstacle avoidance and global asymptotic stability under switching communication topology in the form of
\begin{enumerate}
    \item Synchronization: The tracking errors, $\delta^{k}$ in Eqn.~\eqref{delta_errors}, for all follower agents, $i \in \mathcal{N}$ are cooperatively uniformly ultimately bounded (CUUB).
    
    \item Stability: The origin  is globally stable for all the $k^\text{th}$ order weighted synchronization errors~Eqn.~\eqref{errordyn}, i.e., $e_i^{k,\sigma(t)}$ for all $k\in \{1,\cdots,n\}$, $i\in \{1,\cdots,N\}$, the weighted stability error~Eqn.~\eqref{slidem}, and the NN weight errors \mbox{\( \tilde{\theta} = \theta - \hat{\theta}\)}, \( \tilde{\theta}_{0} = \theta_{0} - \hat{\theta}_{0}\), \(  \tilde{\theta}_{w} = \theta_{w} - \hat{\theta}_{w}\).
    
    \item Collision/Obstacle Avoidance: Agents avoid collisions and obstacles while maintaining synchronization and stability, i.e., $m_{ij}$ in Eqn.~\eqref{col1}, $m_{i0}$ in Eqn.~\eqref{col2}, and $m_{i\mathfrak{c}}$ in Eqn.~\eqref{obs1} remain bounded for all \(t \geq 0\).
    
 \end{enumerate}
\end{theorem}
The CUUB/weight bounds are existential (used for analysis); implementation does not require numerical knowledge of these constants. Boundedness of $\hat\theta$ , $\hat\theta_0$ and $\hat\theta_w$ are enforced by the terms in the NN local tuning laws.
\begin{proof}
\textbf{Part 1:} For a high‐order, high‐dimensional, nonlinear, heterogeneous multi‐agent system, we extend the Lyapunov‐derivative summation and norm‐bounding methodology of \citep{lewis2013} to accommodate time-triggered switching communication topologies, collision and obstacle avoidance, and online disturbance estimation via adaptive neural networks. We consider the Lyapunov function:
\begin{equation}\label{allV}
    V^{\sigma} = V^{\sigma}_1 + V^{\sigma}_2 + V^{\sigma}_3 + V^{\sigma}_4 + V^{\sigma}_5 
\end{equation}
with each term defined to capture the dynamics, NN weight errors, and synchronization errors under the active topology. Specifically \( V^{\sigma}_1 = \frac{1}{2}r^{\top}Pr \), \(V^{\sigma}_2 = \frac{1}{2}\text{tr}\{\tilde{\theta}^{\top}G^{-1}\tilde{\theta}\}  \), \(V^{\sigma}_3 = \frac{1}{2}\text{tr}\{\tilde{\theta}_0^{\top}G_0^{*-1}\tilde{\theta}_0\}    \), \(V^{\sigma}_4 = \frac{1}{2}\text{tr}\{\tilde{\theta}_w^{\top}G_w^{-1}\tilde{\theta}_w\}   \), and \\ \mbox{\(V^{\sigma}_5 = \frac{1}{2}\text{tr}\{E_1 P_1(E_1)^{\top}\} \)}.
Taking the derivative of \( V^{\sigma}_1 \) and using Eqn.~\eqref{errdyn}, yields
\begin{multline} \label{v1dot}
\dot{V}^{\sigma}_1 = r^{\top} P \dot{r} = r^{\top} P \bigg [ \rho - (\nu_{1}L^{\sigma(t)} + \nu_{2}B^{\sigma(t)})(f(x) \\ + u + w - f_0(x_0,t)) \bigg ]. \notag
\end{multline}
Substituting Eqn.~\eqref{controllaw}, we have 
\begin{multline}
\dot{V}^{\sigma}_1 = r^{\top} P \bigg [\rho - (\nu_{1}L^{\sigma(t)} + \nu_{2}B^{\sigma(t)})[{f} \\ + {(D^{\sigma(t)}+B^{\sigma(t)})}^{-1}{\rho} - \hat{\theta}^{\top}{\phi} - \hat{\theta}_{w}^{\top}{\phi_{w}} +\hat{\theta}_0^{\top}{\phi_0} + r - {c} E_0^\top \\ - \Gamma^0 M^0 - \Gamma^1 M^c_I - \Gamma^2 M^c_0 + {w} - {f}_0  ] \bigg].
\end{multline}
Considering Eqn.~\eqref{fxes} through Eqn.~\eqref{fwesh}, Eqn.~\eqref{qdefinite} and invoking the property \( a^{\top} b = \text{tr}\{b \, a^{\top}\} \) after further simplification, we obtain \(\dot{V}^{\sigma}_1\) as
\begin{multline} \label{eq:v1_dot}
    \dot{V}^{\sigma}_1 = -  \frac{1}{2}r^{\top} Q r  - \text{tr}\{\tilde{\theta}^{\top}{\phi}r^{\top} P (D^{\sigma(t)} + B^{\sigma(t)})\} \\ + 
    \text{tr}\{\tilde{\theta}^{\top}{\phi} r^{\top} PA^{\sigma(t)}\} -
    \text{tr}\{\tilde{\theta}_w^{\top}{\phi}_w r^{\top} P(D^{\sigma(t)}+B^{\sigma(t)})\}  \\ + r^{\top} PA^{\sigma(t)}(D^{\sigma(t)} + B^{\sigma(t)})^{-1} \rho + \frac{1}{2}cr^{\top} Q E_0^\top \\ + \text{tr}\{ \tilde{\theta}_w^{\top}{\phi}_w r^{\top} PA^{\sigma(t)} \} + \text{tr}\{ \tilde{\theta}_0^{\top}{\phi}_0 r^{\top} P(D^{\sigma(t)}+B^{\sigma(t)}) \}  \\ - \text{tr} \{ \tilde{\theta}_0^{\top}{\phi}_0 r^{\top} PA^{\sigma(t)} \}  \\ - r^{\top} P(\nu_1L^{\sigma(t)} + \nu_2B^{\sigma(t)})(\varepsilon + \epsilon_w - \epsilon_0) \\ + r^{\top} P(\nu_{1}L^{\sigma(t)} + \nu_{2}B^{\sigma(t)})(\Gamma^0 M^0 + \Gamma^1 M^c_I + \Gamma^2 M^c_0).
\end{multline}
Differentiating \(V^{\sigma}_2 \) with respect to time results in
\begin{equation}
    \dot{V}^{\sigma}_2 = tr\{ \tilde{\theta}^{\top} G^{-1} \dot{\tilde{\theta}} \}. 
\end{equation}
Differentiating Eqn.~\eqref{eq:nn_1} and substituting \(\dot{\tilde{\theta}} = -\dot{\hat{\theta}}\) yields:
\begin{equation}\label{eq:v2_dot}
    \dot{V}^{\sigma}_2 = -tr\{ \tilde{\theta}^{\top} G^{-1} \dot{\hat{\theta}} \}, 
\end{equation}
Likewise, differentiating Eqn.~\eqref{eq:nn_2} and Eqn.~\eqref{eq:nn_3} yields:
\begin{equation}\label{eq:v3_dot}
    \dot{V}^{\sigma}_3 = -tr\{ \tilde{\theta}_0^{\top} G_0^{*-1} \dot{\hat{\theta}}_0 \}, 
\end{equation}
\begin{equation}\label{eq:v4_dot}
    \dot{V}^{\sigma}_4 = -tr\{ \tilde{\theta}_w^{\top} G_w^{-1} \dot{\hat{\theta}}_w \}. 
\end{equation}

The derivative of \(V^{\sigma}_5\) is
\(
    \dot{V}^{\sigma}_5 = tr\{\dot{E}_1P_1(E_1)^{\top}\}. 
\)
Using Eqn.~\eqref{eq:E2_def}, then,
\(
     \dot{V}^{\sigma}_5 = \frac{1}{2}tr\{ E_1 (\bigtriangleup^{\top}P_1 + P_1 \bigtriangleup)(E_1)^{\top}\} + tr\{rl^{\top}P_1(E_1)^{\top}\}.
     \)
Considering Eqn.~\eqref{huwitx}, \(\dot{V}^{\sigma}_5\) becomes
\begin{equation} \label{eq:v5_dot}
     \dot{V}^{\sigma}_5 = -\frac{\beta}{2} || E_1 ||_G^2 + \bar{\alpha}(P_1)||l||||r||||E_1||_G.
\end{equation}
Hence, combining Eqn.~\eqref{eq:v1_dot}, Eqn.~\eqref{eq:v2_dot}, Eqn.~\eqref{eq:v3_dot}, Eqn.~\eqref{eq:v4_dot} and Eqn.~\eqref{eq:v5_dot} into \( \dot{V}^{\sigma} = \dot{V}^{\sigma}_1 + \dot{V}^{\sigma}_2 + \dot{V}^{\sigma}_3 + \dot{V}^{\sigma}_4 + \dot{V}^{\sigma}_5 \) satisfies:
\begin{multline}\label{composite}
    \dot{V}^{\sigma} \le -\bigg [ \frac{1}{2}  \underline{\alpha}(Q) - \frac{\bar{\alpha}(P)\bar{\alpha}(A^{\sigma(t)})}{\underline{\alpha}(D^{\sigma(t)}+B^{\sigma(t)})}||\bar{\lambda}|| \bigg ] ||r||^2 
    \\
    + \bigg[ \frac{\bar{\alpha}(P)\bar{\alpha}(A^{\sigma(t)})}{\underline{\alpha}(D^{\sigma(t)}+B^{\sigma(t)})}||\bigtriangleup||_G||\bar{\lambda}||  + \bar{\alpha}(P_1) \bigg] ||r||||E_1||_G \\
    + \bar{\alpha}(P)\bar{\alpha}(\nu_1L + \nu_2B)T_M||r||
    - \kappa ||\tilde{\theta}||_G^2 -  \kappa_0 ||\tilde{\theta}_0||_{G^*_0}^2 
    \\ - \kappa_w ||\tilde{\theta}_w||_{G_w}^2 + \Phi_n \bar{\alpha}(P)\bar{\alpha}(A^{\sigma(t)})||\tilde{\theta}||_G||r|| 
    \\ + \Phi_{nw} \bar{\alpha}(P)\bar{\alpha}(A^{\sigma(t)})||\tilde{\theta}_w||_{G_w}||r|| \\ + \Phi_{n0} \bar{\alpha}(P)\bar{\alpha}(A^{\sigma(t)})||\tilde{\theta}_0||_{G^*_0}||r||  \\ +  \bar{\alpha}(P)\bar{\alpha}(\nu_1L + \nu_2B)T_N||r|| -\frac{\beta}{2} || E_1 ||_G^2 + \kappa \Theta_n||\tilde{\theta}||_G \\ + \kappa_w \Theta_{nw}||\tilde{\theta}_w||_{G_w} 
    + \kappa_0 \Theta_{n0}||\tilde{\theta}_0||_{G^*_0} + \frac{1}{2} c E_0 \underline{\alpha}(Q)||r||.
\end{multline}
In the interest of keeping this paper concise, we have abbreviated the full Lyapunov‐derivative summation calculation that leads to Eqn.~\eqref{composite}. Eqn.~\eqref{composite} is simply the result of adding the derivative bounds for $V_1$–$V_5$ and then using Cauchy–Schwarz, eigenvalue, and Assumption \ref{assumption1} bounds to control every trace and cross-term. Concretely, the $-\tfrac12\,r^\top Pr$ term and its worst-case coupling bounds produce the leading $\|r\|^2$ coefficient; the remaining consensus and potentials combine with the $\dot V_5$ estimate to yield the $\|r\|\|E_1\|_G$ group; substituting the adaptation laws into $\dot V_2$–$\dot V_4$ generates the $-\kappa\|\tilde\theta\|^2$, $-\kappa_0\|\tilde\theta_0\|^2$, and $-\kappa_w\|\tilde\theta_w\|^2$ terms along with their $\Phi\|\tilde\theta\|\|r\|$ and $\kappa\,\Theta\|\tilde\theta\|$ residuals; and invoking $\|\phi_i\|\le\Phi_n$, $\|\phi_0\|\le\Phi_{n0}$, $\|\phi_{iw}\|\le\Phi_{nw}$ ensures all remaining cross-terms can be absorbed into the defined $\Phi_n,\Theta_n,\Phi_{n0},\Theta_{n0},\Phi_{nw},\Theta_{nw}$ constants.

We rewrite Eqn.~\eqref{composite} in the simplified form,
\begin{equation} \label{modv}
    \dot{V}^{\sigma}= -V_z^{\sigma}(z) \le -z^{\top}Kz + \omega^{\top}z 
\end{equation}
with
\begin{equation*}
K = \begin{bmatrix}
    \frac{\beta}{2} & 0 & 0 & 0 & g \\
    0 & \kappa & 0 &  0 & \gamma_1 \\
    0 & 0 & \kappa_w & 0 & \gamma_2 \\
    0 & 0 & 0 & \kappa_0 & \gamma_3 \\
    g & \gamma_1 & \gamma_2 & \gamma_3 & \mu_1 \\
    \end{bmatrix},
\end{equation*}
\begin{align*}
    \omega &= \bigg [ 0,\kappa \Theta_n,\kappa_w \Theta_{nw}, \kappa_0 \Theta_{n0}, \Lambda \bigg]^{\top}, \\
     \Lambda &= \bar{\alpha}(P)\bar{\alpha}(\nu_1L + \nu_2B)(T_M+T_N) + \mu_2 , \\
    \mathfrak{h} &= \frac{\bar{\alpha}(P)\bar{\alpha}(A^{\sigma(t)})}{\underline{\alpha}(D^{\sigma(t)}+B^{\sigma(t)})}||\bar{\lambda}|| , \\
    \gamma_1 &=  -\frac{1}{2}\Phi_n\bar{\alpha}(P)\bar{\alpha}(A^{\sigma(t)}), \\ 
    \gamma_2  &=  -\frac{1}{2}\Phi_{nw}\bar{\alpha}(P)\bar{\alpha}(A^{\sigma(t)}), \\
    \gamma_3 &=  -\frac{1}{2}\Phi_{n0}\bar{\alpha}(P)\bar{\alpha}(A^{\sigma(t)}), \\ 
    \mu_1  &=  \frac{1}{2} \underline{\alpha}(Q) - \mathfrak{h} , \\
    \mu_2 &= \frac{1}{2} c E_0 \underline{\alpha}(Q), \\
    \mathbf{g}  &=  -\frac{1}{2} \bigg [ \frac{\bar{\alpha}(P)\bar{\alpha}(A^{\sigma(t)})}{\underline{\alpha}(D^{\sigma(t)}+B^{\sigma(t)})}||\bigtriangleup||_G||\bar{\lambda}|| + \bar{\alpha}(P_1) \bigg],
\end{align*}
\begin{equation}\label{eq:z_vector}
       z = \bigg[ ||E_1||_G, ||\tilde{\theta}||_G, ||\tilde{\theta}_w||_{G_w}, ||\tilde{\theta}_0||_{G^*_0}, ||r|| \bigg]^{\top}.
\end{equation}
\(V_z^{\sigma}(z)\) is positive definite whenever the following two conditions are met; \(K\) is positive definite and \(||z|| > \frac{||\omega||}{\underline{\alpha}(K)}\).

\noindent For the \textit{\textbf{first condition}} to be satisfied, according to Sylvester's criterion, we must ensure 
\( 
    \beta > 0,
\) 
\( 
    \frac{\beta}{2}\kappa > 0,
\) 
\( 
    \frac{\beta}{2}\kappa \kappa_w > 0,
\) 
\( 
    \frac{\beta}{2}\kappa \kappa_w \kappa_0 > 0,
\) 
and \( 
    \frac{\beta}{2}\kappa \kappa_w (\kappa_0  \mu_1 - \gamma_3^2) - \frac{\beta}{2}\kappa \gamma_2^2 
    \kappa_0 - \frac{\beta}{2}\gamma_1^2 \kappa_w  \kappa_0 - \mathbf{g}^2 \kappa \kappa_w \kappa_0 > 0,  
\)
with the isolation of \( \mu_1 \), yields
\begin{multline*}
    \mu_1 >
    \frac{1}{\frac{\beta}{2}\kappa \kappa_w \kappa_0} \Bigg[ \frac{\beta}{2}\kappa \kappa_w\gamma_3^2 + \frac{\beta}{2}\kappa \gamma_2^2 \kappa_0  + \frac{\beta}{2}\gamma_1^2 \kappa_w  \kappa_0 \\ +  \mathbf{g}^2 \kappa \kappa_w \kappa_0 \Bigg]
\end{multline*}
For the \textit{\textbf{second condition}} to be satisfied, we require \( ||z|| > B_d \), where
\begin{equation}
    B_d = \frac{||\omega||_1}{\underline{\alpha}(K)}
\end{equation}
with \( \omega = \bigg [ 0,\kappa \Theta_n,\kappa_w \Theta_{nw}, \kappa_0 \Theta_{n0}, \Lambda \bigg]^{\top} \) and \( ||\omega||_1 = \kappa \Theta_n + \kappa_w \Theta_{nw} + \kappa_0 \Theta_{n0} + \Lambda  \). Thus,
\begin{equation}
    B_d = \frac{\kappa \Theta_n + \kappa_w \Theta_{nw} + \kappa_0 \Theta_{n0} + \Lambda}{\underline{\alpha}(K)}
\end{equation}
If \( ||z|| > B_d \), then \(\dot{V}^{\sigma} \leq -V_z^{\sigma}(z)\) with \(V_z^{\sigma}(z)\) being positive definite. 
Singular values are defined for all matrix dimensions. Using the singular value bounds:

   $
     {\underline\sigma(P)}\,\| r\|^2 
     \;\le\;
     r^T P r
     \;\le\;
     {\bar\sigma(P)}\,\| r\|^2,
   $

   $
     \frac1{\bar\sigma(G)}\,\|\tilde \theta\|_G^2 
     \;\le\;
     \mathrm{tr}(\tilde \theta^T G^{-1}\tilde \theta)
     \;\le\;
     \frac1{\underline\sigma(G)}\,\|\tilde \theta\|_G^2,
   $

   $
     \frac1{\bar\sigma(G_0^{*})}\,\|\tilde \theta_0\|_G^2 
     \;\le\;
     \mathrm{tr}(\tilde \theta_0^T G_0^{*-1}\tilde \theta_0)
     \;\le\;
    \frac1{\underline\sigma(G_0^{*})}\,\|\tilde \theta_0\|_G^2,
   $

   $
     \frac1{\bar\sigma(G_w)}\,\|\tilde \theta_w\|_G^2 
     \;\le\;
     \mathrm{tr}(\tilde \theta_w^T G_w^{-1}\tilde \theta_w)
     \;\le\;
    \frac1{\underline\sigma(G_w)}\,\|\tilde \theta_w\|_G^2,
   $
   
   $
     \underline\sigma(P_1)\,\|E_1\|_G^2
     \;\le\;
     \mathrm{tr}(E_1^T P_1 E_1)
     \;\le\;
     \bar\sigma(P_1)\,\|E_1\|_G^2.
   $
Combining and halving gives
$
\tfrac{{\underline\sigma(P)}\,}2\,\| r\|^2 
\;+\;
\frac1{2\,\bar\sigma(G)}\,\|\tilde \theta\|_G^2
\;+\;
\frac1{2\,\bar\sigma(G_0^{*})}\,\|\tilde \theta_0\|_G^2
\;+\;
\frac1{2\,\bar\sigma(G_w)}\,\|\tilde \theta_w\|_G^2
\;+\;
\tfrac{\underline\sigma(P_1)\,}2\,\|E_1\|_G^2
     \;\le\;
     V^\sigma
     \;\le\;
\tfrac{{\bar\sigma(P)}\,}2\,\| r\|^2
\;+\;
\frac1{2\,\underline\sigma(G)}\,\|\tilde \theta\|_G^2
\;+\;
\frac1{2\,\underline\sigma(G_0^{*})}\,\|\tilde \theta_0\|_G^2
\;+\;
\frac1{2\,\underline\sigma(G_w)}\,\|\tilde \theta_w\|_G^2
\;+\;
\tfrac{\bar\sigma(P_1)\,}2\,\|E_1\|_G^2.
$
Considering Eqn.~\eqref{eq:z_vector}, the inequality simplifies to
\begin{equation} \label{zdef}
    \underline{\alpha}(\eta)||z||^2 \leq V^{\sigma} \leq \bar{\alpha}(\zeta)||z||^2,
\end{equation}
where
\(
\eta
\;=\;
\mathrm{diag}\Bigl(\tfrac{\underline\alpha(P)}{2},\tfrac{1}{2\bar\alpha(G)},\frac1{2\,\bar\sigma(G_0^{*})}\,,\tfrac{1}{2\bar\alpha(G_w)},\tfrac{\underline\alpha(P_1)}{2}\Bigr)
\)
and
\(
\zeta
\;=\;
\mathrm{diag}\Bigl(\tfrac{\bar\alpha(P)}{2},\tfrac{1}{2\underline\alpha(G)},\frac1{2\,\underline\sigma(G_0^{*})}\,,\tfrac{1}{2\underline\alpha(G_w)},\tfrac{\bar\alpha(P_1)}{2}\Bigr)
\).
Then for any initial \( z(t_0) \), there exists \( T_0 \) such that
\begin{equation} \label{zref}
    ||z(t)|| \leq \sqrt{\frac{\bar{\alpha}(\zeta)}{\underline{\alpha}(\eta)}}B_d, \quad \forall t \geq t_0 + T_0.
\end{equation}

Let \( \mathcal{Z} = \text{min}_{{||z||}\ge B_d} V^{\sigma}_z(z)\), then
\begin{equation}
    T_0 = \frac{V^{\sigma}(t_0) - \bar{\alpha}(\zeta)(B_d)^2}{\mathcal{Z}}.
\end{equation}
This implies \(r(t)\) is ultimately bounded, which yields \(e_i(t)\), \(e^n(t)\), and \(\delta^k\) CUUB, thus achieving synchronization over~\(\mathcal{G}\).

\textbf{Part 2:} We show that $x_i(t)$ is bounded \( \forall t \ge t_0\) and \( \forall i \in \mathcal{N}\). From Eqn.~\eqref{modv}:

\begin{equation} \label{cond2}
    \dot{V}^{\sigma} \le -\underline{\alpha}(K)||z||^2 + ||\omega||||z||.
\end{equation}
which, together with \(Eqn.~\eqref{zdef}\), we obtain
\begin{equation}
    \frac{d}{dt}(\sqrt{V^{\sigma}}) \le - \frac{\underline{\alpha}(K)}{2\bar{\alpha}(\zeta)}\sqrt{V^{\sigma}} + \frac{||\omega||}{2\sqrt{\underline{\alpha}(\eta)}}.
\end{equation}
Thus, it follows from \citep[Corollary 1.1]{CongWang2004}, that \(V^{\sigma}(t)\) remains bounded for all \(t \ge t_0\). This results, together with Assumption \ref{assp}, yields \(\delta_{i}^k = (x_{i}^k - \psi_{i}^k) - (x_{0}^k - \psi_{0}^k)\) being~CUUB, i.e., \(||x_i|| \le X_{\mathrm{init}}\) and \(||x_0|| \le X_{0, {\mathrm{init}}}\) are bounded~$\forall t \ge t_0$.
\end{proof}

\subsection{Proposed Average Dwell‐Time Requirement}\label{sec:tau_av}
To guarantee exponential convergence despite arbitrary switching, we compute two key scalars ($\mu$ the maximum jump‐factor of the Lyapunov function at any topology switch, and $\varrho_0$ the minimum continuous‐time decay rate of that Lyapunov function under each fixed topology) from the system matrices and then enforce their interplay via a minimum average‐dwell‐time bound. We choose the design parameter \(\bar\lambda\) in Eqn.~\eqref{eq:desparam} so that the error‐polynomial Eqn.~\eqref{eq:hurwitz_poly} is Hurwitz. $P$ and $Q$ the positive-definite matrices from Eqn.~\eqref{pdefinite} and Eqn.~\eqref{qdefinite}, and \(G,G_0^*,G_w\) the NN‐adaptation gains of Section~\ref{NN_Laws_Defn}. First, from our Lyapunov‐derivative analysis in Section~\ref{leader_follower_cons} under any topology \(\sigma\) we extract the worst‐case decay rate
\begin{equation} \label{eq:RhoZeroDefinition}
    \varrho_0 = \frac{1}{2}\underline{\alpha}(Q) - \frac{\bar{\alpha}(P)\bar{\alpha}(A^{\sigma(t)})}{\underline{\alpha}(D^{\sigma(t)}+B^{\sigma(t)})}\|\bar{\lambda}\|.
\end{equation}
Second, at each switch \(t_s\) the Lyapunov function may jump by at most a factor \(\mu\), determined by
\begin{equation} \label{eq:MuDefinition}
 \mu = \max \frac{\bar{\alpha}(\zeta)}{\underline{\alpha}(\eta)}. 
\end{equation}
where the diagonal entries of
\[
\eta
\;=\;
\mathrm{diag}\Bigl(\tfrac{\underline\alpha(P)}{2},\tfrac{1}{2\bar\alpha(G)},\frac1{2\,\bar\sigma(G_0^{*})}\,,\tfrac{1}{2\bar\alpha(G_w)},\tfrac{\underline\alpha(P_1)}{2}\Bigr)
\]
and
\[
\zeta
\;=\;
\mathrm{diag}\Bigl(\tfrac{\bar\alpha(P)}{2},\tfrac{1}{2\underline\alpha(G)},\frac1{2\,\underline\sigma(G_0^{*})}\,,\tfrac{1}{2\underline\alpha(G_w)},\tfrac{\bar\alpha(P_1)}{2}\Bigr)
\]
from Eqn.~\eqref{zdef} feed directly into the jump factor \(\mu\). Requiring that each switch’s worst‐case growth \(\mu\) be more than offset, on average, by the continuous decay yields the standard minimum average‐dwell‐time condition.
Using the values of $\varrho_0$ in Eqn.~\eqref{eq:RhoZeroDefinition} and $\mu$ in Eqn.~\eqref{eq:MuDefinition}, we adopt the minimum average dwell time condition Eqn.~\eqref{tau_average},
\begin{equation}\label{eq:dwelltime}
\tau_a^* > -\frac{\ln \mu}{\ln (1-\varrho_0)}.  
\end{equation}
Unlike previous dwell time frameworks \citep{Hespanha, Liberzon2003}, we derive expressions for the jump factor $\mu$ and decay rate $\varrho_0$ in a heterogeneous multiagent setting that simultaneously handles online NN‐based cancellation of unknown dynamics, collision/obstacle avoidance, and switching communication topologies. Embedding these tailored expressions into our adaptive, safety‐constrained, time‐based switching framework extends previous dwell time works to this MAS context.

\begin{proof}
Let \(\{t_s\}_{s=0}^{\infty}\) be the strictly increasing sequence of switching instants, where \(t_0 \ge 0\). At each \(t_s\), the active topology switches to \(\bar{\mathcal{G}}^{\sigma(t_s)}\). Assume for each topology \(\bar{\mathcal{G}}^{{\bar{c}}}\), there exists a continuously differentiable Lyapunov function \(V^{m}(x)\) (e.g., Eqn.~\eqref{allV}) that measures the consensus error of the agents under that topology. For brevity, at each switching instant \(t_s\), we write \(\sigma = \sigma(t_s)\). Thus, \(\mathcal{G}^{\sigma}\) will denote the active topology at \(t_s\). Following the discrete‐time decay condition of \citep[Theorem 2.19]{Rawlings_Mayne}, we evaluate the Lyapunov function across each switching instant. Specifically, we define
\begin{multline}
\Delta V^{\sigma}(x(t_s)) = V^{\sigma}(x(t_{s+1})) - V^{\sigma}(x(t_s)) \\ \leq -\varrho_{\sigma} V^{\sigma}(x(t_s)),
\end{multline}
assuming there exists \(\varrho_{\sigma} > 0\) (decay rate for the active topology \(\bar{\mathcal{G}}^{\sigma(t_s)}\)), we obtain
\begin{equation}\label{lyapfuncdec0}
V^{\sigma}(x(t_{s+1})) \leq (1 - \varrho_{\sigma}) V^{\sigma}(x(t_s)),
\end{equation} 
We define \(\varrho_0\) as the worst-case decay rate—the smallest \(\varrho_{\sigma}\) among all topologies.
\(
   \varrho_0 = \min_{\sigma \in \{1, \dots, M\}} \varrho_{\sigma},
\)
with the assumption \(\varrho_0 > 0\). This ensures that even under the least favorable topology, the Lyapunov function still decays. In particular, the per‐step decay factor is at least \((1 - \varrho_{0})\), so to guarantee stability irrespective of the active topology, we write
\begin{equation}\label{lyapfuncdec1}
V^{\sigma}(x(t_{s+1})) \leq (1 - \varrho_0) V^{\sigma}(x(t_{s-})).
\end{equation}
The derivative of the composite Lyapunov function Eqn.~\eqref{composite} is lower‐bounded by
\[
\dot{V}^{\sigma} \le -\left[\frac{1}{2}\underline{\alpha}(Q) - \frac{\bar{\alpha}(P)\bar{\alpha}(A^{\sigma(t)})}{\underline{\alpha}(D^{\sigma(t)}+B^{\sigma(t)})} \|\bar{\lambda}\| \right] \|r\|^2,
\]
where the first negative term is the decay term, re-written as
\[
\varrho_0 = \frac{1}{2}\underline{\alpha}(Q) - \frac{\bar{\alpha}(P)\bar{\alpha}(A^{\sigma(t)})}{\underline{\alpha}(D^{\sigma(t)}+B^{\sigma(t)})} \|\bar{\lambda}\|.
\]
At each switching instant, \(t_s\), the Lyapunov function may increase by at most a factor \(\mu \ge 1\). If the topology switches from \(\bar{\mathcal{G}}^{\sigma(t_{s-})}\) to \(\bar{\mathcal{G}}^{\sigma(t_s)}\), then 
\begin{equation}\label{increase_at_switch} V^{\sigma(t_s)}(x(t_s)) \leq \mu V^{\sigma(t_{s-})}(x(t_{s-})). \end{equation}
The composite Lyapunov function is bounded as given by Eqn.~\eqref{zdef}. At each switching instant, the system switches from a topology \(\sigma(t_s)\) with bounds \((\underline{\alpha}(\eta^{\sigma(t_s)}),\,\bar{\alpha}(\zeta^{\sigma(t_s)}))\) to another topology \(\sigma(t_{s-})\) with bounds \((\underline{\alpha}(\eta^{\sigma(t_{s-})}),\,\bar{\alpha}(\zeta^{\sigma(t_{s-})}))\), then the jump in the Lyapunov function is bounded by
\[
V^{\sigma(t_s)}(x(t_s)) \le \frac{\bar{\alpha}(\zeta^{\sigma(t_s)})}{\underline{\alpha}(\eta^{{\sigma(t_{s-})}})}\,V^{\sigma(t_{s-})}(x(t_{s-})).
\]
Thus, the jump‐factor is

\begin{equation} 
 \mu = \max \frac{\bar{\alpha}(\zeta)}{\underline{\alpha}(\eta)}. 
\end{equation}

Between successive switching times \(t_0\) and \(t_1\), the Lyapunov function decays at most by a factor \((1-\varrho_0)\), so 
\[
V^{\sigma(t_{1-})}\bigl(x(t_{1-})\bigr) \;\le\; \bigl(1 - \varrho_0\bigr)\,V^{\sigma(t_0)}\bigl(x(t_0)\bigr).
\]
Combining these inequalities at \(t_1\) gives
\[
V^{\sigma(t_1)}\bigl(x(t_1)\bigr) 
   \;\le\; \mu\,\bigl(1 - \varrho_0\bigr)\,V^{\sigma(t_0)}\bigl(x(t_0)\bigr). 
\]
Similarly, over the next interval \([t_1,\,t_2)\) and at \(t_2\), we obtain
 \[   V^{\sigma(t_2)}\bigl(x(t_2)\bigr) 
   \;\le\; \mu^2\,\bigl(1 - \varrho_0\bigr)^2\,V^{\sigma(t_0)}\bigl(x(t_0)\bigr). 
 \] 
Hence, each switching instant introduces at most a factor of \(\mu\), while each interval between switches provides a decay factor of \((1-\varrho_0)\).

Applying this recursively from \(s=1\) to \(s=N^{\sigma}(0,t)\) and noting that \(\sum_{s=1}^{N^{\sigma}(0,t)}(t_s - t_{s-1}) = t - t_0\) yields
\[
  V^{\sigma(t)}\bigl(x(t)\bigr)
  \;\le\;
  \bigl[\mu\,(1-\varrho_0)\bigr]^{N^{\sigma}(0,t)}\,
  V^{\sigma(t_0)}\bigl(x(t_0)\bigr).
\]
\(\{t_0,\, t_1,\, t_2,\, \dots,\, t_{N^{\sigma}(0,t)}\}\) are the switching instants and \( N^{\sigma}(t_0, t) \) is the number of switches over the interval \([t_0,t]\).
The number of switches in \([t_0,t]\) is bounded by \(N^{\sigma}(t_0, t) \) given by Eqn.~\eqref{numberswitches0}. In our weighted version, we define an average dwell time \(\tau_a^*\) that reflects the topology connectivity quality, giving
\begin{equation}\label{numswitches}
N^{\sigma}(0,t) \le N_0 + \frac{t}{\tau_a^*}.
\end{equation}
Rather than handling each \(V^{\sigma}(x(t))\) separately, we define the composite Lyapunov function
\begin{equation}\label{comlyap}
V(x(t)) = \max_{\sigma\in\mathcal{M}}\left\{ \iota^{\sigma(t)} V^{\sigma(t)}(x(t)) \right\},
\end{equation}
where \(\iota^{\sigma(t)}\) depends on the algebraic connectivity of the graph $\bar{\mathcal{G}}^{\sigma(t)}$ defined as the second-smallest eigenvalue of its Laplacian matrix $L^{\sigma}$, denoted as \(s_2(L^{\sigma})\). We choose
\[
\iota^{\sigma(t)}
\;=\;\frac{1}{1 + \gamma\,\bigl(\kappa_{\max}-s_2(L^{\sigma})\bigr)},
\]
with \(\gamma>0\).
Substituting the bound Eqn.~\eqref{numswitches} into Eqn.~\eqref{comlyap} yields:
\begin{align*}
  V(x(t))
  \;\le\; \\
  \Bigl(\max_\sigma\iota^\sigma\Bigr)\,
  \mu^{N_0}(1-\varrho_0)^{N_0}\,
  \Bigl[\mu^{\tfrac{1}{\tau_a^*}}(1-\varrho_0)\Bigr]^{t}\,
  V(x(0)).
\end{align*}
Hence, to guarantee \(V(x(t))\to0\) as \(t\to\infty\) we require
\[
  \mu^{\frac1{\tau_a^*}}\,(1-\varrho_0)\;<\;1,
\]
which rearranges to the average‐dwell‐time condition
\begin{equation} \label{tau_average}
  \tau_a^* \;>\; -\frac{\ln\mu}{\ln(1-\varrho_0)}.
\end{equation} 
\end{proof}

\newpage

\begin{proposition} \label{thm:BoundedInputs}
For a system with follower dynamics~\eqref{planta} and leader dynamics~\eqref{plantref} with $n \geq 3$, and under Assumptions \ref{assp}(1) and \ref{assumption1}, the control law \eqref{eq:ProposedControlLaw} generates a uniformly bounded input under the active communication topology, i.e., there exists finite bounds $u_n<\infty$ such that
\(
\|u(t)\|\le u_n,
\)
\(
\forall t\ge t_0.
\)

Moreover, the inter-agent, agent–leader, and agent–obstacle separations satisfy 
$\|x_i^1(t)-x_j^1(t)\|\ge \psi_{ij}^1>0,
\|x_i^1(t)-x_0^1(t)\|\ge \psi_{i0}^1>0,
\|x_i^1(t)-O_{\mathfrak{c}}\|\ge R>0$, for all $t\ge t_0$.
\end{proposition}

\begin{proof}
    See Appendices~\ref{app:HOCBF} and \ref{app:BoundedInputs}.
\end{proof}

\section{Numerical Example} \label{NUMERICALEXAMPLE}
We consider a leader-follower formation with one leader and five followers, each with distinct nonlinear dynamics and external disturbances. In our 2D example we set $n=3$ and $p=2$, so that $\,x_i^1,x_i^2,x_i^3\in \mathbb {R}^2$ represent the position, velocity and acceleration vectors, respectively. Component‐wise, we write $x_{i}^{1} = \bigl[x_{i}^{1,1},\,x_{i}^{1,2}\bigr]^\top$, $x_{i}^{2} = \bigl[x_{i}^{2,1},\,x_{i}^{2,2}\bigr]^\top$ and $x_{i}^{3} = \bigl[x_{i}^{3,1},\,x_{i}^{3,2}\bigr]^\top$.
Figure \ref{fig:1} shows the information flow structure of the multiagent system under a switching communication topology (switches every 5 seconds). In our simulations we employed a 12-dimensional LIP basis for both the state and disturbance approximators.  
For the state‐dependent basis $\phi_i(x_i)$ we used
$\displaystyle
\phi_i(x_i)=\bigl[\,1,\;x_i^{1,1},\;x_i^{1,2},\;x_i^{2,1},\;x_i^{2,2},\;x_i^{3,1},\;x_i^{3,2},\;(x_i^{1,1})^2,\;(x_i^{1,2})^2, \\ \;(x_i^{2,1})^2,\;(x_i^{2,2})^2,\;(x_i^{3,1})^2\bigr]^\top
$
i.e.\ a constant term, the six state components, and the squares of the first five components (through the first acceleration entry).  For the time‐dependent disturbance basis $\phi_{iw}(t)$ we retained
$\displaystyle
\phi_{iw}(t)
=\bigl[
1,\;
\sin t,\;
\cos t,\;
\sin2t,\;
\cos2t,\;
\sin3t,\;
\cos3t,\;
\\ \sin^2t,\;
\cos^2t,\;
\sin t\,\cos t,\;
e^{-t},\;
t\,e^{-t}
\bigr]^\top.
$
This combination of polynomial, trigonometric, and exponential terms provides a rich yet compact dictionary for approximating unknown smooth nonlinearities and time-varying disturbances \citep{stevenbrunton2016}, keeps the adaptation law computationally tractable, and ensures each basis function is locally Lipschitz (as required in our stability analysis).
The dynamics of the leader agent are described by the following equations:

\noindent\textbf{Leader:}
\begin{align*}
\dot{x}_{0}^{1,1} = x_{0}^{2,1}, \quad \dot{x}_{0}^{1,2} = x_{0}^{2,2}; \quad
\dot{x}_{0}^{2,1} = x_{0}^{3,1}, \quad \dot{x}_{0}^{2,2} = x_{0}^{3,2};\\
\dot{x}_{0}^{3,1} = -\,x_{0}^{2,1}\;-\;2\,x_{0}^{3,1}\;+\;1\;+\;3\sin\bigl(2\pi\,x_{0}^{1,1}\bigr)\\
\quad -\;\tfrac{1}{3}\bigl(x_{0}^{1,1}+x_{0}^{2,1}-1\bigr)^{2}\,\bigl(x_{0}^{1,1}+4\,x_{0}^{2,1}+3\,x_{0}^{3,1}-1\bigr),\\
\dot{x}_{0}^{3,2} = -\,x_{0}^{2,2}\;-\;2\,x_{0}^{3,2}\;+\;1\;+\;3\sin\bigl(2\pi\,x_{0}^{1,2}\bigr)\\
\quad -\;\tfrac{1}{3}\bigl(x_{0}^{1,2}+x_{0}^{2,2}-1\bigr)^{2}\,\bigl(x_{0}^{1,2}+4\,x_{0}^{2,2}+3\,x_{0}^{3,2}-1\bigr).
\end{align*}
For each follower agent \(i\) (\(i = 1, \dots, 5\)), the common dynamics are:
\begin{align*}
\dot x_{i}^{1,1} = x_{i}^{2,1}, \quad
\dot x_{i}^{1,2} = x_{i}^{2,2}; \quad
\dot x_{i}^{2,1} = x_{i}^{3,1}, \quad
\dot x_{i}^{2,2} = x_{i}^{3,2};
\end{align*}

\noindent\textbf{Agent 1:}  
\begin{align*}
\dot x_{1}^{3,1} = -\,x_{1}^{2,1}\sin\bigl(x_{1}^{2,1}\bigr)
  -\cos^2\bigl(x_{1}^{3,1}\bigr)
  -0.1\,(x_{1}^{2,1})^2
  \\ -0.05\,(x_{1}^{3,1})^2
  +u_{1}^{1}
  +w_{x_1},\\
\dot x_{1}^{3,2} = -\,x_{1}^{2,2}\sin\bigl(x_{1}^{2,2}\bigr)
  -\cos^2\bigl(x_{1}^{3,2}\bigr)
  -0.1\,(x_{1}^{2,2})^2
  \\-0.05\,(x_{1}^{3,2})^2
  +u_{1}^{2}
  +w_{y_1}.
\end{align*}

\noindent\textbf{Agent 2:}  
\begin{align*}
\dot x_{2}^{3,1} = -\,(x_{2}^{2,1})^2
  +0.01\,x_{2}^{3,1}
  -0.01\,(x_{2}^{2,1})^3
  \\-0.1\,(x_{2}^{3,1})^2
  -0.1\,x_{2}^{2,1}
  +u_{2}^{1}
  +w_{x_2},\\
\dot x_{2}^{3,2} = -\,(x_{2}^{2,2})^2
  +0.01\,x_{2}^{3,2}
  -0.01\,(x_{2}^{2,2})^3
  \\-0.1\,(x_{2}^{3,2})^2
  -0.1\,x_{2}^{2,2}
  +u_{2}^{2}
  +w_{y_2}.
\end{align*}

\noindent\textbf{Agent 3:}  
\begin{align*}
\dot x_{3}^{3,1} = x_{3}^{2,1}
  +\sin\bigl(x_{3}^{3,1}\bigr)
  -0.05\,(x_{3}^{2,1})^2
  -0.05\,(x_{3}^{3,1})^2
  \\+u_{3}^{1}
  +w_{x_3},\\
\dot x_{3}^{3,2} = x_{3}^{2,2}
  +\sin\bigl(x_{3}^{3,2}\bigr)
  -0.05\,(x_{3}^{2,2})^2
  -0.05\,(x_{3}^{3,2})^2
  \\+u_{3}^{2}
  +w_{y_3}.
\end{align*}

\noindent\textbf{Agent 4:}  
\begin{align*}
\dot x_{4}^{3,1} &= 
  -3\,(x_{4}^{1,1}+x_{4}^{2,1}-1)^{2}\,(x_{4}^{1,1}+x_{4}^{2,1}+x_{4}^{3,1}-1)
  \\ &-x_{4}^{2,1}
  -x_{4}^{3,1} 
  +0.5\sin\bigl(2\pi x_{4}^{1,1}\bigr)
  +\cos\bigl(2\pi x_{4}^{1,1}\bigr)
 \\ & +u_{4}^{1}
  +w_{x_4},\\
\dot x_{4}^{3,2} &=
  -3\,(x_{4}^{1,2}+x_{4}^{2,2}-1)^{2}\,(x_{4}^{1,2}+x_{4}^{2,2}+x_{4}^{3,2}-1)
  \\ &-x_{4}^{2,2}
  -x_{4}^{3,2}
  +0.5\sin\bigl(2\pi x_{4}^{1,2}\bigr)
  +\cos\bigl(2\pi x_{4}^{1,2}\bigr)
 \\ & +u_{4}^{2}
  +w_{y_4}.
\end{align*}

\noindent\textbf{Agent 5:}  
\begin{align*}
\dot x_{5}^{3,1} = -\,x_{5}^{2,1}
  -0.05\,(x_{5}^{3,1})^2
  +x_{5}^{1,1}
  +w_{x_5},\\
\dot x_{5}^{3,2} = -\,x_{5}^{2,2}
  -0.05\,(x_{5}^{3,2})^2
  +u_{5}^{2}
  +w_{y_5}.
\end{align*}

\begin{figure}[ht]
    \centering
    \begin{tikzpicture}[>=stealth, scale=0.90]

    \node[node_style] (0) at (-1,2) {0};
    \node[node_style] (1) at (0,2.5) {1};
    \node[node_style] (2) at (1,2) {2};
    \node[node_style] (3) at (1,1) {3};
    \node[node_style] (4) at (0,0.5) {4};
    \node[node_style] (5) at (-1,1) {5};

    \draw[edge_style] (0) -- (1);
    \draw[edge_style_1] (1) -- (2);
    \draw[edge_style_1] (1) -- (5);
    \draw[edge_style_1] (1) -- (3);
    \draw[edge_style_1] (2) -- (3);
    \draw[edge_style_1] (4) -- (5);
    \draw[edge_style] (0) -- (5);
    \draw[edge_style_1] (5) -- (3);
    \draw[edge_style_1] (2) -- (4);

    \node at (0, 3.25) {\(\bar{G}^{\sigma(t_1)}\)};

    \node[node_style] (0b) at (3,2) {0};
    \node[node_style] (1b) at (4,2.5) {1};
    \node[node_style] (2b) at (5,2) {2};
    \node[node_style] (3b) at (5,1) {3};
    \node[node_style] (4b) at (4,0.5) {4};
    \node[node_style] (5b) at (3,1) {5};

    \draw[edge_style] (0b) -- (1b);
    \draw[edge_style] (1b) -- (2b);
    \draw[edge_style_1] (3b) -- (2b);
    \draw[edge_style] (0b) -- (5b);
    \draw[edge_style_1] (1b) -- (4b);
    \draw[edge_style_1] (1b) -- (5b);
    \draw[edge_style_1] (5b) -- (4b);
    \draw[edge_style_1] (2b) -- (4b);
    \draw[edge_style] (5b) -- (3b);

    \node at (4,3.25) {\(\bar{G}^{\sigma(t_2)}\)};

    \node[node_style] (0d) at (3,-2) {0};
    \node[node_style] (1d) at (4,-1.5) {1};
    \node[node_style] (2d) at (5,-2) {2};
    \node[node_style] (3d) at (5,-3) {3};
    \node[node_style] (4d) at (4,-3.5) {4};
    \node[node_style] (5d) at (3,-3) {5};

    \draw[edge_style] (0d) -- (1d);
    \draw[edge_style] (0d) -- (5d);
    \draw[edge_style_1] (1d) -- (5d);
    \draw[edge_style_1] (1d) -- (2d);
    \draw[edge_style_1] (3d) -- (5d);
    \draw[edge_style_1] (4d) -- (5d);
    \draw[edge_style_1] (2d) -- (3d);
    \draw[edge_style_1] (2d) -- (4d);

    \node at (4,-0.75) {\(\bar{G}^{\sigma(t_3)}\)};

    \node[node_style] (0c) at (-1,-2) {0};
    \node[node_style] (1c) at (0,-1.5) {1};
    \node[node_style] (2c) at (1,-2) {2};
    \node[node_style] (3c) at (1,-3) {3};
    \node[node_style] (4c) at (0,-3.5) {4};
    \node[node_style] (5c) at (-1,-3) {5};

    \draw[edge_style] (0c) -- (1c);
    \draw[edge_style] (0c) -- (5c);
    \draw[edge_style_1] (1c) -- (2c);
    \draw[edge_style_1] (3c) -- (4c);
    \draw[edge_style_1] (1c) -- (5c);
    \draw[edge_style_1] (1c) -- (3c);
    \draw[edge_style_1] (2c) -- (5c);
    \draw[edge_style] (2c) -- (4c);
    \draw[edge_style_1] (2c) -- (3c);
    \draw[edge_style] (5c) -- (4c);


\node (A) [draw=blue, fit= (0) (1) (2) (3) (4) (5), inner sep=-2mm, ultra thick, fill=blue!10, fill opacity=0.2, ellipse] {};
\node [xshift=5ex, rotate=-90, blue] at (A.east) {};

\node (B) [draw=blue, fit= (0b) (1b) (2b) (3b) (4b) (5b), inner sep=-2mm, ultra thick, fill=blue!10, fill opacity=0.2, ellipse] {};
\node [xshift=5ex, rotate=-90, blue] at (A.east) {};

\node (C) [draw=blue, fit= (0c) (1c) (2c) (3c) (4c) (5c), inner sep=-2mm, ultra thick, fill=blue!10, fill opacity=0.2, ellipse] {};
\node [xshift=5ex, rotate=-90, blue] at (A.east) {};

\node (D) [draw=blue, fit= (0d) (1d) (2d) (3d) (4d) (5d), inner sep=-2mm, ultra thick, fill=blue!10, fill opacity=0.2, ellipse] {};
\node [xshift=5ex, rotate=-90, blue] at (A.east) {};

\path[-{Latex[scale=2]}]
(A.north east) edge[bend left=25, sloped, anchor=center, above, blue]  node {} (B.north west);
\path[-{Latex[scale=2]}]
(B.south east) edge[bend left=25, sloped, anchor=center, above, blue]  node {} (D.north east);
\path[-{Latex[scale=2]}]
(D.south west) edge[bend left=25, sloped, anchor=center, above, blue]  node {} (C.south east);
\path[-{Latex[scale=2]}]
(C.north west) edge[bend left=20, sloped, anchor=center, above, blue]  node {} (A.south west);

    \node at (0,-0.75) {\(\bar{G}^{\sigma(t_4)}\)};

    \end{tikzpicture}
    \caption{The considered switching communication topology of the augmented graph \(\bar{\mathcal{G}}\) in Section~\ref{NUMERICALEXAMPLE} illustrating different topologies at switching times \(t_1\) to \(t_4\). Leader is denoted by \(0\) and follower agents 1, 2, 3, 4 and 5}
    \label{fig:1}
\end{figure}
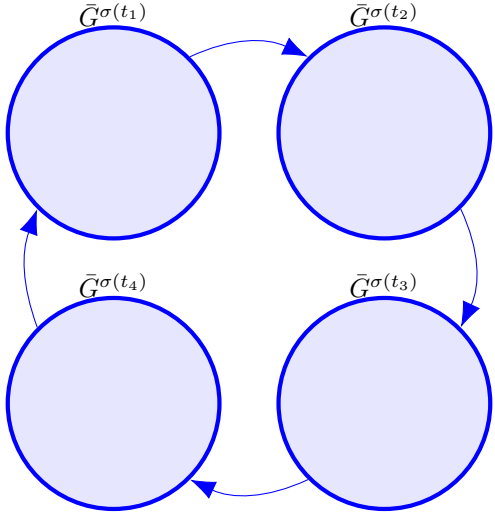

\begin{figure}[htbp]
    \centering
    \begin{subfigure}[b]{0.410\textwidth}
    \centering
    \includegraphics[width=2.9in]
    {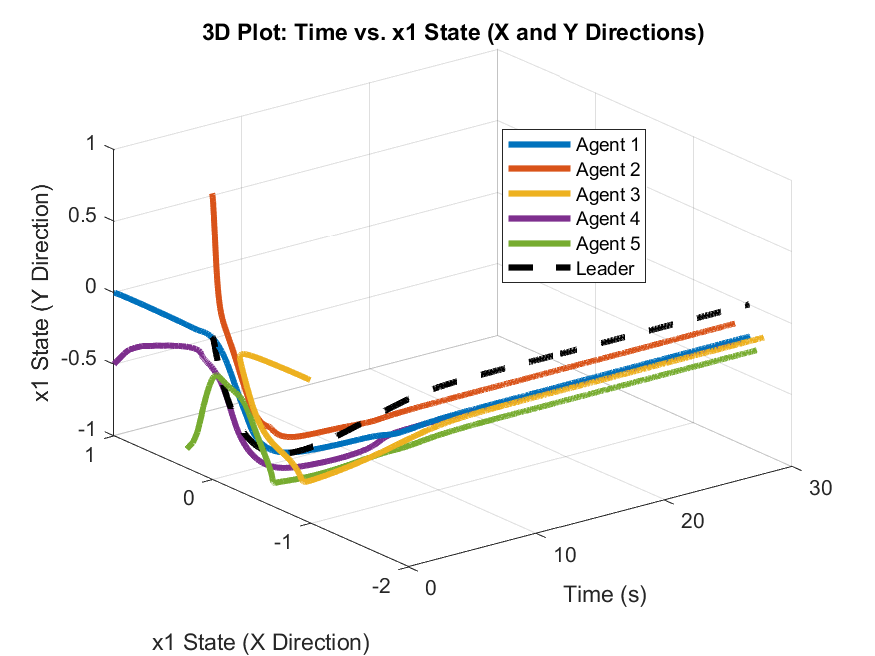}
    \caption{Leader and Follower Agents $x_i^1$ states components}
    \label{fig:2}
    \end{subfigure}
    \hfill
    \begin{subfigure}[b]{0.410\textwidth}
    \centering
    \includegraphics[width=2.9in]{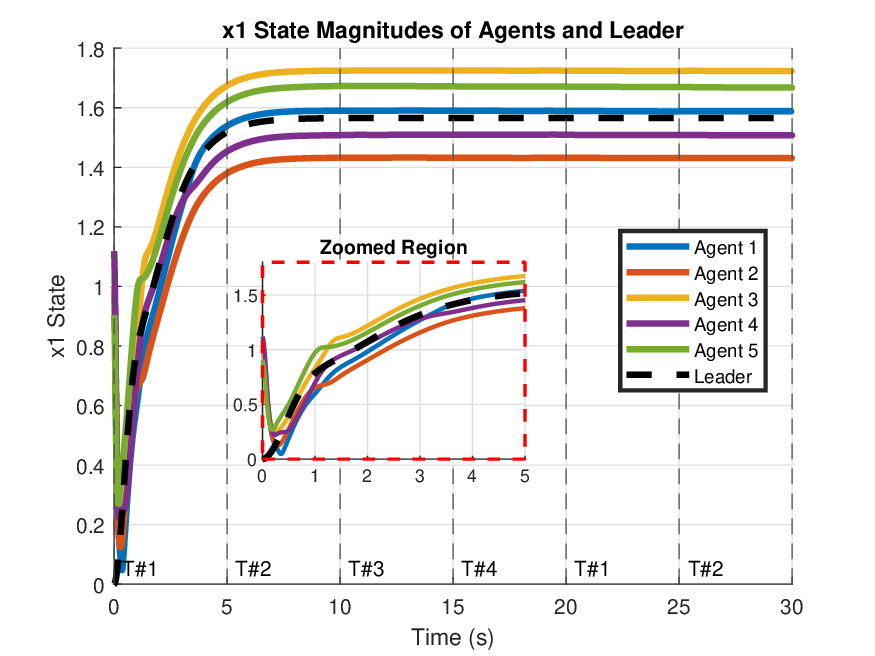}
    \caption{Leader and Follower Agents $x_i^1$ state magnitude}
    \label{fig:3}
    \end{subfigure}
    \caption{The corresponding evolution of the positional states $x_i^1 \in \mathbb{R}^2$, for the \mbox{leader $i=0$}, and follower agents  $i \in \{1,2,\cdots,5\}$}
    \label{fig:combined0}
\end{figure}
\begin{figure}[htbp]
    \begin{subfigure}[b]{0.410\textwidth}
    \centering
    \includegraphics[width=2.9in]{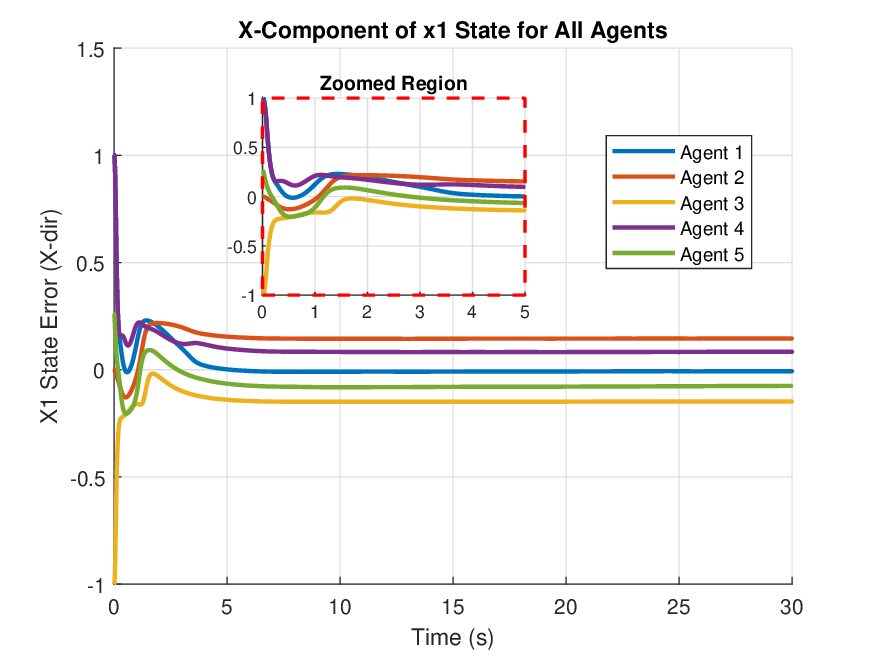}
    \caption{Relative $x_i^1$ state Error between Agents and Leader in x-direction}
    \label{fig:5}
    \end{subfigure}
            \hfill
    \begin{subfigure}[b]{0.410\textwidth}
    \centering
    \includegraphics[width=2.9in]{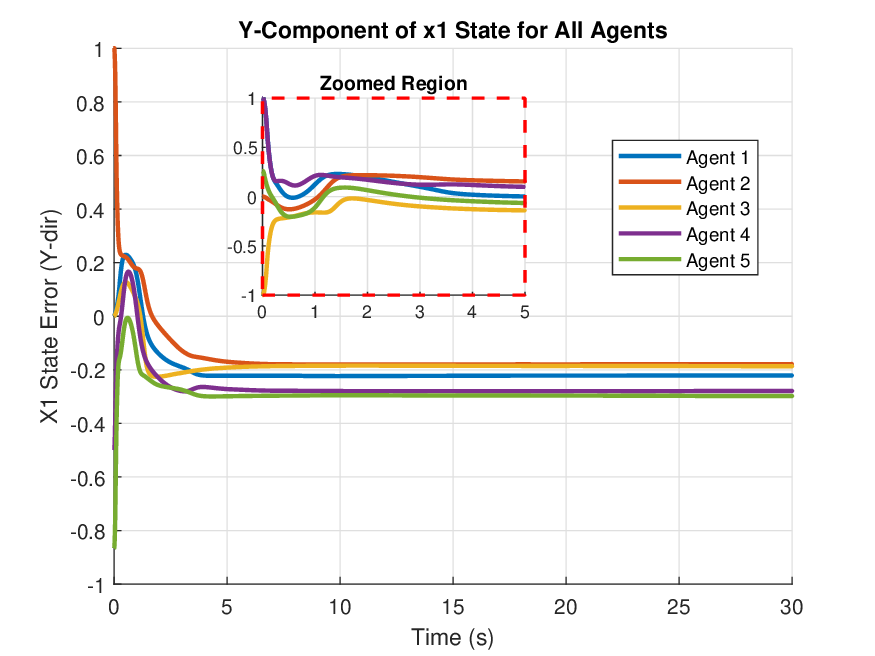}
    \caption{Relative $x_i^1$ state Error between Agents and Leader in y-direction}
    \label{fig:6}
    \end{subfigure}
    \caption{Time evolution of the relative $x_i^1 \in \mathbb{R}^2$ state error (agents vs.\ leader): (a) $x$-direction; (b) $y$-direction.}
    \label{fig:combined1}
\end{figure}
In this numerical example, Agent 0 (leader) guides five followers with unique nonlinear dynamics and external disturbances over alternating communication topologies (see Figure \ref{fig:1}). Switching between graphs impacts dynamics and coordination; however, the system remains globally stable as agents converge toward the leader's trajectory.

Figures \ref{fig:2} and \ref{fig:3} shows significant \(x_i^1\) positional state fluctuations during switching, especially with less connected graphs (e.g., \(\bar{G}^{\sigma(t_3)}\)), highlighting the effect of reduced communication. Nonetheless, the followers stabilize over time. Figure \ref{fig:3} displays the active topology at 5-second intervals, labeled as \(T\#\). Additionally, Figure \ref{fig:5} demonstrates a rapid decline in relative \(x_i^1\) state errors between agents and the leader, confirming effective formation maintenance and tracking despite disturbances. 

In summary, the multiagent system exhibits global stability and convergence to the leader’s trajectory under switching communication topologies, mitigating disturbances and ensuring effective error minimization and control stabilization in dynamic environments.

\section{Conclusion}\label{CONCLUSION}
This paper has presented a distributed adaptive control strategy for uncertain leader-follower multiagent systems operating under time-triggered switchings of communication topologies. By addressing the challenges posed by heterogeneous nonlinear dynamics, unknown parameters in dynamics, external disturbances, and evolving communication structures, the proposed framework ensures adaptive synchronization, formation maintenance, and stability in dynamic environments. The integration of neural network-based estimation techniques with adaptive tuning laws allows agents to effectively handle uncertainties in their dynamics and disturbances, while potential functions guarantee collision and obstacle avoidance.

Using a composite Lyapunov analysis framework that accounts for switching communication topologies, we formally established the stability and convergence properties of the proposed control laws under time-triggered communication transitions. The average dwell-time condition served as a critical constraint for ensuring smooth inter-topology transitions while preserving global asymptotic stability. Numerical simulations further demonstrated the efficacy of the approach, highlighting its ability to achieve coordinated motion and to suppress relative position errors even under reduced connectivity scenarios.

To broaden the applicability and enhance the practicality of our framework, several avenues for future work are apparent. First, while our analysis focuses on systems whose dynamics decompose into integrator chains plus locally Lipschitz nonlinearities, extending the method to encompass rigid‐body models, input delays, or more intricate coupling structures would be valuable. Second, stability relies on an average‐dwell‐time constraint on topology switching; scenarios involving prolonged or arbitrary disconnects fall outside our guarantees. Third, incorporating control barrier functions for collision and obstacle avoidance will guard against potential chattering or oscillations induced by repulsive potentials, ensuring that control inputs remain within hardware limits. Finally, exploring reduced‐order or sparsely parameterized basis function sets for the neural approximators may strike an even better balance between approximation accuracy and computational load, facilitating deployment on resource‐constrained platforms. Further research directions include extensions to event-triggered switching of communication topologies, incorporating learning-based methods for topology optimization, and real-world implementations in applications such as autonomous vehicles, UAV swarms, and robotic teams. Collectively, these enhancements promise to extend our theoretical guarantees into a wider range of real‐world multi‐agent applications while addressing the computational and practical constraints inherent in distributed systems.

\bibliographystyle{asmems4}

\bibliography{asme2e}

\appendix       
\section{Safety Assurance and Invariance Analysis}
\noindent This section provides a rigorous, formal proof of the safety guarantees embedded in the proposed control architecture. The core of the analysis relies on establishing the forward invariance of specifically defined ``safe sets" in the state space. By proving that agent trajectories, once initiated within these sets, will never leave them, we can certify that inter-agent collisions, agent-leader collisions, and agent-obstacle collisions are all systematically avoided for all time, even under the influence of bounded disturbances and switching communication topologies.

\subsection{High-Order Control Barrier Functions}\label{app:HOCBF}

\noindent To formally establish the safety guarantees, we employ the methodology of High-Order Control Barrier Functions (HOCBFs). HOCBFs are a powerful extension of standard Control Barrier Functions (CBFs) designed to handle systems where the control input's effect on the safety-critical output is indirect—that is, for systems with a high relative degree. This tool is exceptionally well-suited for the control-affine  agent dynamics~\eqref{planta}, which is rewritten in the form
\begin{equation}\label{eq:controlaffine1}
   \dot x_i \;=\; F_x(x_i)+F_u \,u_i+F_w \,w_i(t), 
\end{equation}
$
  F_x(x_i)=\begin{bmatrix}x_i^2\\ x_i^3\\ \vdots\\ f_i(x_i)\end{bmatrix},\quad
  F_u=\begin{bmatrix}0\\0\\ \vdots\\ 1\end{bmatrix},\quad
  F_w=\begin{bmatrix}0\\0\\ \vdots\\ 1\end{bmatrix},
$
\\[9pt]
\noindent where, as before, $ x_i = [ x_i^1,  x_i^2, \cdots,  x_i^n ]^\top \in \mathbb{R}^{ n {\cdot {\pd} }}$, $u_i, w_i\in\mathbb R^p,$ and $f_i:\mathcal D\subset\mathbb R^{n \cdot {\pd}}\to\mathbb R^{p}$ is locally Lipschitz. The vector field $F_x:\mathcal D\to\mathbb R^{np}$ is locally Lipschitz (it inherits this from $f_i$), while $F_u, F_w$ are constant matrix Lipschitz continuous with Lipschitz constant 0 and the unknown disturbance satisfies $\|w_i(t)\|\le w_n$ for all $t\ge0$.

\begin{definition}[Relative degree \citep{Xiao2019,Cortez2022}]
A continuously differentiable function smooth $h\in C^{\,\mathfrak{d}}$ is said to have (least) relative degree $\mathfrak{d}$ w.r.t.\ $(F_x,{F_u})$ whenever $L_{F_u} L_{F_x}^{\,\upsilon} h \equiv 0$ for $\upsilon=1,\ldots,\mathfrak{d}-2$ and $L_{F_u} L_{F_x}^{\,\mathfrak{d}-1} h \not\equiv 0$, where $L_{F_x}h:=\nabla h\,F_x$, $L_{F_u}h:=\nabla h\,F_u$, and $L_{F_w}h:=\nabla h\,F_w$ denote Lie derivatives of $h$ with respect to $F_x$, $F_u$ and $F_w$, respectively. We assume the disturbance channel is matched in the same sense, i.e., $L_{F_w} L_{F_x}^{\,\upsilon} h \equiv 0$ for $\upsilon=1,\ldots,\mathfrak{d}-2$ and $
L_{F_w} L_{F_x}^{\,\mathfrak{d}-1} h \not\equiv 0.$ 
\end{definition}
Let the safe set be $\mathcal C:=\{x_i:\,h(x)\ge0\}$ and we define the HOCBF recursion with class-$\mathcal K$ functions $\chi_\upsilon$:
\begin{equation}\label{eq:psi-recursion}
{\vartheta}_0(x_i):=h(x_i),\,\, 
{\vartheta}_{k}(x_i):=\dot{\vartheta}_{\upsilon-1}(x_i)+\chi_{\upsilon}({\vartheta}_{\upsilon-1}(x_i)),
\end{equation}
with $\upsilon=1,\ldots,\mathfrak{d}-1$. We set $\mathcal C_\upsilon:=\{x_i:{\vartheta}_{\upsilon-1}(x_i)\ge0\}$ and define the admissible region 
$\mathcal X_{\text{adm}}:=\bigcap_{\upsilon=1}^{\mathfrak{d}-1}\mathcal C_\upsilon$.
Note that the total derivative along Eqn.~\eqref{eq:controlaffine1} is
$\dot{\vartheta}_\upsilon=L_{F_x}{\vartheta}_\upsilon+L_{F_u}{\vartheta}_\upsilon\,u_i+L_{F_w}{\vartheta}_\upsilon\,w_i$.

\begin{definition}[Forward invariance \citep{Cortez2022}]\label{def:forward-invariance}
Consider the disturbed control–affine system \eqref{eq:controlaffine1}
and the safe set $\mathcal C:=\{x:\,h(x)\ge 0\}$ with the admissible domain
$\mathcal X_{\rm adm}:=\bigcap_{k=1}^{\mathfrak d}\{\,\vartheta_{k-1}(x)\ge 0\,\}$.
We say that $\mathcal C$ is forward invariant under the closed–loop
policy $u(\cdot)$ if, for every initial condition $x(0)\in \mathcal C\cap \mathcal X_{\rm adm}$
and for every measurable disturbance $w(\cdot)$ satisfying $\|w(t)\|\le w_n$ i.e.,
the corresponding solution $x(t)$ exists for all $t\ge 0$ and
\(
x(t)\in \mathcal C, \,\, \forall t\ge 0.
\)
\end{definition}

\begin{assumption}\label{ass:rd-matching}
For the system~\ref{eq:controlaffine1} and a barrier $h\in C^{\mathfrak d}$ of least relative degree $\mathfrak d$ with respect to $(F_x,F_u)$, the disturbance channel $F_w$ is matched in the relative–degree sense, i.e., $
L_{F_w}L_{F_x}^{\,\upsilon}h(x)\ \equiv\ 0,\, \upsilon=0,1,\ldots,\mathfrak d-2,
$. In other words, the disturbance does not appear in the Lie–derivative chain of $h$ before the control does.
Equivalently, $w$ affects $h$ at order $\mathfrak d$, the same order at which $u$ enters.
\end{assumption}

\begin{lemma}\label{lemma:HOCBF}
For the control-affine system \eqref{eq:controlaffine1}, let $h\in C^{\mathfrak d}(\mathbb{R}^n)$ be a control barrier function with relative degree~$\mathfrak d$ with respect to the input channel $F_u$. Define
\[
\vartheta_0(x_i):=h(x_i),\quad 
\vartheta_{\nu}(x_i):=\dot{\vartheta}_{\nu-1}(x_i)+\chi_{\nu}\!\big(\vartheta_{\nu-1}(x_i)\big).
\]

Then under the Assumptions \ref{assp} and \ref{ass:rd-matching}, the safe set $\mathcal{C}:=\{x\in\mathbb{R}^n:\ h(x)\ge 0\}$ is forward invariant
if there exists $u_i\in\mathcal{U}$ such that
\begin{multline}\label{eq:robust-hocbf-compact}
L_{F_x}\vartheta_{\mathfrak d-1}(x_i)\;+\;L_{F_u}\vartheta_{\mathfrak d-1}(x_i)\,u_i-\big\|L_{F_w}\vartheta_{\mathfrak d-1}(x_i)\big\|w_n\\
+\chi_{\mathfrak d}\!\big(\vartheta_{\mathfrak d-1}(x_i)\big)\ \ge\ 0,
\end{multline}
for all $x_i\in\mathcal{X}_{\rm adm}$.
Moreover, on the boundary $\{x\in\mathbb{R}^n:\ \vartheta_0(x)=\cdots=\vartheta_{\mathfrak d-1}(x)=0\}$, \eqref{eq:robust-hocbf-compact} reduces to
\begin{equation}\label{eq:boundary-hocbf}
L_{F_x}^{\mathfrak d} h(x_i)+L_{F_u} L_{F_x}^{\mathfrak d-1} h(x_i)u_i-\big\|L_{F_w} L_{F_x}^{\mathfrak d-1} h(x_i)\big\|w_n \ge 0.
\end{equation}
\end{lemma}

\begin{proof}
    See \cite{Xiao2019,Cortez2022,Xiao2022}.
\end{proof}

\begin{corollary}\label{corollary:HOCBF}
Under the conditions of Lemma~\ref{lemma:HOCBF}, there exist explicit gain conditions for the input~\eqref{eq:ProposedControlLaw} that render the safety sets
\begin{align*}
\mathcal S_{ij} &:=\{\,x:\|x_i^{1}-x_j^{1}\|\ge \psi_{ij}\,\},
\\
\mathcal S_{i0} &:=\{\,x:\|x_i^{1}-x_0^{1}\|\ge \psi_{i0}\,\},
\\
\mathcal S_{i\mathfrak c} &:=\{\,x:\|x_i^{1}-\Omega\|\ge \psi_{i\mathfrak c}\,\}
\end{align*}
forward invariant, thereby safety with both collision and obstacle avoidance is achieved.

\end{corollary}

\begin{proof}
Due to the similarity of the analysis, we only prove that the set $\mathcal S_{ij}$  (no collisions with other agents) is forward invariant under our control law Eqn.~\eqref{eq:ProposedControlLaw1} with bounded disturbances. Demonstrations of forward invariance of the sets $\mathcal S_{i0}$ (no collisions with leader) and $\mathcal S_{i\mathfrak c}$ (no collisions with obstacle) follow the same procedure.
We define
\begin{align*}
h_{ij}(x)&:= \|x_i^{1}-x_j^{1}\|-\psi_{ij},\quad
\ell_{ij}:=x_i^{1}-x_j^{1},\\
& s_{ij}:=\|\ell_{ij}\|,\quad
n_{ij}:=\frac{\ell_{ij}}{s_{ij}}\ (s_{ij}>0).
\end{align*}
The barrier is \(h_{ij}(x)=s_{ij}-\psi_{ij}\), so \(h_{ij}\) depends only on the position blocks
\(x_i^{1},x_j^{1}\) with constant \(\psi_{ij}\). Hence all other partials (w.r.t.\ \(x_i^{k},\,k\ge2\)) are zero. Therefore,
\begin{align*}
ds_{ij} &= \tfrac{1}{2}(\ell_{ij}^\top \ell_{ij})^{-1/2}\,d(\ell_{ij}^\top \ell_{ij})
   = \frac{1}{\|\ell_{ij}\|}\,\ell_{ij}^\top\,d\ell_{ij},\\
d\ell_{ij} &= d x_i^{1}-d x_j^{1},\\
\Rightarrow\quad
ds_{ij} &= n_{ij}^\top\,d x_i^{1}\;-\;n_{ij}^\top\,d x_j^{1}.
\end{align*}
By the definition of the gradient (as the row mapping differentials to \(ds\)),
\( \frac{\partial s}{\partial x_i^{1}}=n_{ij}^\top,
        \frac{\partial s}{\partial x_j^{1}}=-\,n_{ij}^\top
\)
and
\(
\frac{\partial h_{ij}}{\partial x_i^{1}}=n_{ij}^\top,
        \frac{\partial h_{ij}}{\partial x_j^{1}}=-\,n_{ij}^\top.
\)
Since $dh_{ij} = ds_{ij}$ (\(\psi_{ij}\) is constant),   $dx_{i}=\dot x_i^1=x_i^2$ and $dx_{j}=\dot x_j^1=x_j^2$, then
\begin{equation}\label{eq:Lf-h}
L_{F_x} h_{ij}=\dot h_{ij}=n_{ij}^\top(x_i^2-x_j^2).
\end{equation}
Differentiating $n_{ij}=\ell_{ij}/s_{ij}$ to get
\begin{align}
\dot s_{ij}&=\frac{\ell_{ij}^\top(x_i^2-x_j^2)}{s_{ij}}=n_{ij}^\top(x_i^2-x_j^2), \notag \\
\dot n_{ij}&=\frac{1}{s_{ij}}\Big(I-n_{ij}n_{ij}^\top\Big)(x_i^2-x_j^2).
\end{align}
Differentiating Eqn.~\eqref{eq:Lf-h} once more,
\begin{align}
L_{F_x}^2 h_{ij}&=\ddot h_{ij} \nonumber\\
&=\dot n_{ij}^\top(x_i^2-x_j^2)+n_{ij}^\top(x_i^3-x_j^3)\nonumber\\
&=\frac{1}{s_{ij}}(x_i^2-x_j^2)^\top\!\Big(I-n_{ij}n_{ij}^\top\Big)(x_i^2-x_j^2) \nonumber \\ &+n_{ij}^\top(x_i^3-x_j^3).
\label{eq:Lf2-h}
\end{align}
Here $\big(I-n_{ij}n_{ij}^\top\big)$ is the orthogonal projector onto the subspace tangent to the safety sphere; the first term of Eqn.~\eqref{eq:Lf2-h} is nonnegative and encodes the rotation of $n_{ij}$.
For $k=1,\ldots,n-1$, we get the general form
\begin{equation}\label{eq:Lfk-h-pattern}
L_{F_x}^{k} h_{ij}\;=\; n_{ij}^\top\!\big(x_i^{k+1}-x_j^{k+1}\big)\;+\;O_{ij}^{(k)}(x),
\end{equation}
where each $O_{ij}^{(k)}(x)$ collects the terms built from $n_{ij},\dot n_{ij},\ddot n_{ij},\ldots$ times lower-order differences $x_i^k-x_j^k$ also written in a finite sum of products of the form
$(\partial^q n_{ij}/\partial t^q) \cdot(x_i^k-x_j^k)$.
Equation \eqref{eq:Lf2-h} is the base case $k=2$. The step $k\to k+1$ follows by differentiating Eqn.~\eqref{eq:Lfk-h-pattern}, applying the product rule to $n_{ij}$ and using $\dot x_i^{k+1}=x_i^{k+2}$.

On the operating set used in the Lyapunov stability analysis in Parts 1–2 earlier, all states are bounded and $s_{ij}\ge \psi_{ij}>0$. Therefore there exists a finite constant $M_{ij}$ such that
\begin{equation}
 |O_{ij}^{(n-1)}(x)|\ \le\ M_{ij} 
\end{equation}
Because $h_{ij}$ depends only on positions, the system has relative degree $n$ w.r.t. each agent’s input/disturbance. For $k=1,\dots,n-1$, then
$
L_{F_u}L_{F_x}^{k-1} h_{ij}\equiv 0, L_{F_w} L_{F_x}^{k-1} h_{ij}\equiv 0,
$
since $F_u$ and $F_w$ act only in the $x^n$ coordinates, but $L_{F_x}^{k}h_{ij}$ depends at most on $x^{1:k+1}$.
At $k=n-1$ we get a term that is in $x_i^n,x_j^n$:
\begin{align}
L_{F_x}^{\,n-1} h_{ij}
= n^\top(x_i^{n}-x_j^{n}) + O_{ij}^{(n-1)}(x).
\end{align}
At order $n$ using Eqn.~\eqref{eq:controlaffine1}, we get
\begin{align}
x_i^{n}-x_j^{n}
=(F_x-{F_x}_j)
+(u_i-u_j)
+(w_i-w_j).
\end{align}
Therefore using the boundary condition Eqn.~\eqref{eq:boundary-hocbf},
\begin{equation}\label{eq:bnd-HOCBF}
L_{F_x}^{\,n} h_{ij} + L_{F_u} L_{F_x}^{\,n-1} h_{ij}\,u - \big\|L_d L_{F_x}^{\,n-1}h_{ij}\big\|\,{w_n,}_{ij}\ \ge\ 0,
\end{equation} we get
\begin{align}\label{eq:nth-der}
&L_{F_x}^{\,n}h_{ij} + L_{F_u} L_{F_x}^{n-1} h_{ij}\,u + \|L_{F_w} L_{F_x}^{\,r-1}h(x)\|\,w_n\; \notag\\
&= n_{ij}^\top(F_x-{F_x}_j) \notag \\
&+ n_{ij}^\top(u_i-u_j) +  n_{ij}^\top(w_i-w_j) +  {O}^{(n-1)}_{ij} \ge 0.
\end{align}
Specifically,
\begin{align}\label{eq:driftchannel}
    L_{F_x}^{\,n}h_{ij} &= n_{ij}^\top(F_x-{F_x}_j)+  O_{ij}^{(n-1)}(x)
\end{align} 
\begin{align}\label{eq:inputchannel}
    L_{F_u} L_{F_x}^{\,n-1}h_{ij}\,u &= n_{ij}^\top(u_i-u_j) 
\end{align}
\begin{align}\label{eq:disturbancechannel} 
L_{F_w} L_{F_x}^{\,n-1}h_{ij} &= n_{ij}^\top(w_i-w_j)
\end{align}
On the boundary $s_{ij}=\psi_{ij}>0$ and within the bounded set guaranteed by Lyapunov analysis in Part~1-2 above, there exists a finite bound
\begin{equation}\label{eq:Dnom}
\big|\,L_{F_x}^{\,n}h_{ij}\,\big|\ \le\ \mathfrak{D}^{\mathrm{nom}}_{ij}.
\end{equation}
Using the boundary condition Eqn.~\eqref{eq:boundary-hocbf} and $\|n_{ij}\|=1$, Eqn.~\eqref{eq:nth-der} becomes the scalar inequality
\begin{align}\label{eq:bnd-scalar}
&n_{ij}^\top(u_i-u_j)\ \ge\ {w_n,}_{ij} - L_{F_x}^{\,n} h_{ij}, \\
&{w_n,}_{ij}\ \ge\ \sup_{t}\big|\,n_{ij}^\top(w_i-w_j)\,\big|. \notag
\end{align}
Splitting Eqn.~\eqref{eq:ProposedControlLaw} into the collision repulsive parts between agent $i$ and $j$, agent $i$ and leader $0$, and agent $i$ and obstacle $\mathfrak{c}$,
\begin{align*}
    u_i^0 = \Gamma^0_{\mathfrak{c}}  m_{i\mathfrak{c}} n_{i\mathfrak{c}}, \,   
     u_i^{c_1} = \Gamma^1_{ij}  m_{ij} n_{ij},   \,  
    u_i^{c_2} =  \Gamma^2_{i0}  m_{i0} n_{i0}.
\end{align*}
The rest of the non-repulsive term is $u_i^d,$
\begin{align}
       u_i^d &= \frac{\rho_i}{d_i^{\sigma(t)}+b_{i0}^{\sigma(t)}} - \hat{\theta}_i^{\top}{\phi}_i - \hat{\theta}_{iw}^{\top}{\phi_{iw}} +\hat{\theta}_0^{\top}{\phi_0} + r_i \notag \\ & - {c}_i E_{i0}^\top. 
\end{align} 
Focusing on the collision avoidance repulsive term $u_i^{c_1}$ and the rest of the non-repulsive term $u_i^d$, the RHS of Eqn.~\eqref{eq:bnd-scalar} becomes
\begin{align}\label{eq:proj-raw}
n_{ij}^\top(u_i-u_j)
&= n_{ij}^\top(u_i^{\mathrm{col}}-u_j^{\mathrm{col}})
   + n_{ij}^\top(u_i^{\mathrm{rest}}-u_j^{\mathrm{rest}})\nonumber\\
&=\Gamma^1_{ij}\,n_{ij}^\top\!\Big(\sum_{q\ne i}m_{iq}n_{iq}-\sum_{q\ne j}m_{jq}n_{jq}\Big) \nonumber
   \\ &+ n_{ij}^\top(u_i^{\mathrm{rest}}-u_j^{\mathrm{rest}})\nonumber\\
&=\Gamma^1_{ij}\,n_{ij}^\top\!\Big(\underbrace{m_{ij}n_{ij}}_{q=j}
+\sum_{q\ne i,j}m_{iq}n_{iq}
\nonumber \\ &-\underbrace{m_{ji}n_{ji}}_{q=i}
-\sum_{q\ne j,i}m_{jq}n_{jq}\Big)
\nonumber \\ &+ n_{ij}^\top(u_i^{\mathrm{rest}}-u_j^{\mathrm{rest}}).
\end{align}
Using $n_{ji}=-n_{ij}$ and $m_{ji}=m_{ij}$,
\[
n_{ij}^\top(m_{ij}n_{ij})-n_{ij}^\top(m_{ji}n_{ji})=m_{ij}+m_{ij}=2m_{ij}.
\]
Hence
\begin{align}\label{eq:pair-plus-rem}
n_{ij}^\top(u_i-u_j)
&=2\Gamma^1_{ij}\,m_{ij}
\notag \\ & +\Gamma^1_{ij}\,n_{ij}^\top\!\Big(\sum_{q\ne i,j}m_{iq}n_{iq}-\sum_{q\ne j,i}m_{jq}n_{jq}\Big)
\notag \\ &+ n_{ij}^\top(u_i^{\mathrm{rest}}-u_j^{\mathrm{rest}}).
\end{align}
At the pair boundary $s_{ij}=\psi_{ij}$, $m_{ij}=\varpi/\psi_{ij}$; by Cauchy--Schwarz,
\begin{multline}\label{eq:col_bounded}
n_{ij}^\top\!\Big(\sum_{q\ne i,j}m_{iq}n_{iq}-\sum_{q\ne j,i}m_{jq}n_{jq}\Big) \\ \ge\ -\Big\|\sum_{q\ne i,j}m_{iq}n_{iq}-\sum_{q\ne j,i}m_{jq}n_{jq}\Big\|\  \ge\ -\mathfrak{R}_{ij},
\end{multline}
for some finite $\mathfrak{R}_{ij}$. Likewise, we bound
\begin{align}\label{eq:rest_bounded}
\big|\,n_{ij}^\top(u_i^{\mathrm{rest}}-u_j^{\mathrm{rest}})\,\big|\ \le\ \mathfrak{E}_{ij}.
\end{align}
Therefore, from Eqn.~\eqref{eq:pair-plus-rem},
\begin{equation}\label{eq:lb-input}
n_{ij}^\top(u_i-u_j)\ \ge\ 2\Gamma^1_{ij}\,\frac{\varpi}{\psi_{ij}}\ -\ \Gamma^1_{ij} \mathfrak{R}_{ij}\ -\ \mathfrak{E}_{ij}.
\end{equation}
Substituting Eqn.~\eqref{eq:lb-input} and Eqn.~\eqref{eq:Dnom} into Eqn.~\eqref{eq:bnd-scalar}. A sufficient condition is
\[
2\Gamma^1_{ij}\,\frac{\varpi}{\psi_{ij}} - \Gamma^1_{ij} \mathfrak{R}_{ij} - \mathfrak{E}_{ij}\ \ge\ {w_n,}_{ij}+\mathfrak{D}^{\mathrm{nom}}_{ij},
\]
which is equivalent to the explicit gain rule
\begin{equation}\label{eq:Gamma-rule-final}
\Gamma^1_{ij}\;>\;\frac{\big({w_n,}_{ij}+\mathfrak{E}_{ij}+\mathfrak{D}^{\mathrm{nom}}_{ij}\big)\,\psi_{ij}}{\,2\varpi - \mathfrak{R}_{ij}\,\psi_{ij}\,}\,,
\quad 2\varpi>\mathfrak{R}_{ij}\psi_{ij}\quad
\end{equation}
(when only the pair $(i,j)$ is active near the boundary, $\mathfrak{R}_{ij}=0$ so the denominator is $2\varpi$).
Under Eqn.~\eqref{eq:Gamma-rule-final}, the boundary condition Eqn.~\eqref{eq:bnd-scalar} holds, and by the HOCBF invariance result, $\mathcal S_{ij}$ is forward invariant.

\begin{remark}[Leader and obstacle terms]
We define $h_{i0}=\|x_i^1-x_0^1\|-\psi_{i0}$ and $h_{i{\mathfrak{c}}}=\|x_i^1-O_{\mathfrak{c}}\|-R$ with repulsion
$u_i^0 = \Gamma^0_{\mathfrak{c}}  m_{i\mathfrak{c}} n_{i\mathfrak{c}} \,   
$ and $
    u_i^{c_2} =  \Gamma^2_{i0}  m_{i0} n_{i0}.$
Repeating the previous proof with $(j)$ replaced by $(0)$ or $(w)$ yields identical gain rules for $\Gamma_0,\Gamma_2$ (with the substitutions of bounds), ensuring $\|x_i^1-x_0^1\|\ge\psi_{i0}$ and $\|x_i^1-x_{\mathfrak{c}}^1\|\ge\psi_{i\mathfrak{c}}$ are forward invariant.
\end{remark}
\end{proof}

\subsection{Boundedness of the Control Input} \label{app:BoundedInputs}

\noindent We now present the proof for Proposition~\ref{thm:BoundedInputs}
\begin{proof}[Proof of Proposition~\ref{thm:BoundedInputs}]
By forward invariance established in Corollary \ref{corollary:HOCBF}, the separations,
$\|x_i^1(t)-x_j^1(t)\|\ge \psi_{ij}^1>0,
\|x_i^1(t)-x_0^1(t)\|\ge \psi_{i0}^1>0,
\|x_i^1(t)-O_{\mathfrak{c}}\|\ge R>0,$ hold  $\forall t\ge t_0$. Consequently, the repulsion scalars obey uniform bounds
\(
0\le m_{ij}(t)\le M_{ij}:=\frac{\varpi}{\psi_{ij}^1},\quad
0\le m_{i0}(t)\le M_{i0}:=\frac{\varpi}{\psi_{i0}^1},\quad
0\le m_{i{\mathfrak{c}}}(t)\le M_{i{\mathfrak{c}}}:=\Bigl(\frac{R^2-(\mathfrak{L}+\varepsilon)^2}{R^2-\mathfrak{L}^2}\Bigr)^{\!2}.
\)
The repulsive components of Eqn.~\eqref{eq:ProposedControlLaw} are linear combinations of $\{m_{ij}\}$, $\{m_{i0}\}$, and $\{m_{i{\mathfrak{c}}}\}$ with fixed gains (e.g., $\Gamma_1,\Gamma_0,\Gamma_2$). Since the number of neighbors/obstacles per agent is finite under the admissible switching, those sums are uniformly bounded; thus each repulsive term (agent–agent, leader, obstacle) is bounded.

For the non-repulsive (regulation/adaptive) part $u_i^d$, Assumption~\ref{assumption1} ensures all basis functions and parameter estimates are uniformly bounded. By Theorem~\ref{thm:unified}, $r_i$ and $E_{i0}$ are ultimately bounded. Moreover, by Assumption~\ref{assp1} there exists $\underline d>0$ such that
\(
d_i^{\sigma(t)}+b_{i0}^{\sigma(t)}\ \ge\ \underline d \quad\Rightarrow\quad
\big(d_i^{\sigma(t)}+b_{i0}^{\sigma(t)}\big)^{-1}\le \underline d^{-1},
\)
so any factor involving this inverse is uniformly bounded. Hence $\|u_i^d(t)\|$ is uniformly bounded.

Each component $u_i^d$, the agent–agent repulsion, the leader repulsion, and the obstacle repulsion is bounded by a finite constant; summing yields a finite bound for $\|u_i(t)\|$. Stacking over $i$ gives $\|u(t)\|\le u_n$ for some $u_n<\infty$, $\forall t\ge t_0$. Because the degree ($d_i^{\sigma(t)}$) and leader-coupling ($b_{i0}^{\sigma(t)}$) terms remain bounded under the admissible switching signal, the bound is preserved across topology switches.
\end{proof}

\end{document}